\pdfoutput=1
\RequirePackage{ifpdf}
\ifpdf 
\documentclass[pdftex]{sigma}
\else
\documentclass{sigma}
\fi

\numberwithin{equation}{section}

\newtheorem{Theorem}{Theorem}[section]
\newtheorem*{Theorem*}{Theorem}

\newtheorem{Proposition}[Theorem]{Proposition}
 { \theoremstyle{definition}

\newtheorem{Reduction}[Theorem]{Reduction}

\newtheorem{Remark}[Theorem]{Remark} }

\newcommand{\beqa}{\begin{eqnarray}}
\newcommand{\eeqa}{\end{eqnarray}}
\newcommand{\ba}{\begin{eqnarray*}}
\newcommand{\ea}{\end{eqnarray*}}
\newcommand{\bA}{\boldsymbol{A}}
\newcommand{\bB}{\boldsymbol{B}}
\newcommand{\bC}{\boldsymbol{C}}
\newcommand{\bD}{\boldsymbol{D}}

\newcommand{\bG}{\boldsymbol{G}}

\newcommand{\bL}{\boldsymbol{L}}
\newcommand{\bM}{\boldsymbol{M}}

\newcommand{\bT}{\boldsymbol{T}}
\newcommand{\bU}{\boldsymbol{U}}
\newcommand{\bV}{\boldsymbol{V}}
\newcommand{\bW}{\boldsymbol{W}}

\newcommand{\bPsi}{\boldsymbol{\Psi}}
\newcommand{\bPhi}{\boldsymbol{\Phi}}

\begin{document}

\newcommand{\arXivNumber}{2410.21878}

\renewcommand{\PaperNumber}{050}

\FirstPageHeading

\ShortArticleName{B\"acklund--Darboux Transformations for Super KdV Type Equations}

\ArticleName{B\"acklund--Darboux Transformations\\ for Super KdV Type Equations}

\Author{Lingling XUE~$^{\rm a}$, Shasha WANG~$^{\rm a}$ and Qing Ping LIU~$^{\rm b}$}

\AuthorNameForHeading{L.~Xue, S.~Wang and Q.P.~Liu}

\Address{$^{\rm a)}$~Department of Applied Mathematics, Ningbo University, Ningbo 315211, P.R.~China}
\EmailD{\href{mailto:xuelingling@nbu.edu.cn}{xuelingling@nbu.edu.cn}, \href{mailto:wang99sha@163.com}{wang99sha@163.com}}

\Address{$^{\rm b)}$ Department of Mathematics, China University of Mining and Technology,\\
\hphantom{$^{\rm b)}$}~Beijing 100083, P.R.~China}
\EmailD{\href{mailto:qpl@cumtb.edu.cn}{qpl@cumtb.edu.cn}}

\ArticleDates{Received October 30, 2024, in final form June 26, 2025; Published online July 03, 2025}

\Abstract{By introducing a Miura transformation, we derive a generalized super modified Korteweg--de Vries (gsmKdV) equation from the generalized super KdV (gsKdV) equation. It is demonstrated that, while the gsKdV equation takes Kupershmidt's super KdV (sKdV) equation and Geng--Wu's sKdV equation as two distinct reductions, there are also two equations, namely Kupershmidt's super modified KdV (smKdV) equation and Hu's smKdV equation, which are associated with the gsmKdV equation. By analyzing the flows within the gsKdV and gsmKdV hierarchies, we specifically derive the first negative flows associated with both hierarchies.
We then construct a number of B\"acklund--Darboux transformations (BDTs) for both the gsKdV and gsmKdV equations, elucidating the interrelationship between them. By proper reductions, we are able not only to recover the previously known BDTs for Kupershimdt's sKdV and smKdV equations, but also to obtain the BDTs for the Geng--Wu's sKdV/smKdV and Hu's smKdV equations. As applications, we construct some exact solutions for those equations. Since all flows of the sKdV or smKdV hierarchy share the same spatial parts of spectral problem, thus these Darboux matrices and spatial parts of BTs are applicable to any flow of those hierarchies.}

\Keywords{Darboux transformations; B\"acklund transformations; Miura transformations; super KdV equations}

\Classification{35Q53; 35A30; 58J72}

\section{Introduction}

 The celebrated Korteweg--de Vries (KdV) equation, which models solitary waves in shallow water, is one of most important integrable systems. It was related to a similar nonlinear partial differential equation, i.e., modified KdV (mKdV) equation, via the Miura transformation.
 Both KdV and mKdV equations are integrable and their integrable properties such as Lax pairs, multi-soliton solutions, infinitely many symmetries and conserved quantities, bi-Hamiltonian structures, and solvability by the inverse scattering transformation, are well known. With the development of theory of integrable systems, numerous classical systems including KdV and mKdV equations are extended to super or supersymmetric (see \cite{bc,CK,bfg,bfgk,yyg,gw,xbl, GS,ggg,gg,mg1,mg2,hp,Hr,hu,IK,skr,ppk,kuper1,kuper4,MR,mathieu,sNLS,kai,kw,lz,dfz} and references therein).

Various different super extended versions of KdV/mKdV equations exist in the literature. Kupershmidt was the first to conduct such study and in 1984 he proposed a super KdV (sKdV) equation with one bosonic field and one fermionic field, and a super mKdV (smKdV) equation together with the Miura transformation and Lax pairs \cite{kuper1}. Holod and Pakuliak in 1989 generalized Kupershmidt's sKdV equation and considered a generalized sKdV (gsKdV) equation which involves one bosonic field and two fermionic fields \cite{hp}. In~1997, Hu proposed another smKdV equation of one bosonic field and two fermionic fields, and also gave its spectral problem~\cite{hu}. In fact, Hu's smKdV equation is related to the gsKdV equation of Holod--Pakuliak via a Miura transformation, so the former is a modification of the latter~\cite{zty}.
Geng and Wu in 2010 proposed a new sKdV equation of one bosonic field and one fermionic field and constructed its spectral problem, bi-Hamiltonian structures, and infinite conservation laws \cite{gw}.
Also, replacing fermionic fields by the so-called ren-fields, Lou extended the Holod--Pakuliak's gsKdV equation to a ren-KdV type equation and gave the corresponding spectral problem~\cite{lou}.

The super systems mentioned above have been studied and their various properties have been established. The symmetries for Kupershmidt's sKdV/smKdV equations were constructed in~\cite{ker,kg}.
B\"acklund--Darboux transformations (BDTs) were established to generate associated integrable discrete systems for super nonlinear Schr\"odinger equations \cite{gg} and gsKdV equations~\cite{xl,zhou}. They were also employed to construct Grassman extensions of Yang--Baxter maps~\cite{pa,ggg} and then Grassman extended discrete integrable systems from Yang--Baxter maps~\cite{sk,skr}. Aguirre et al.\ in 2018 investigated the integrability of the supersymmetric mKdV hierarchy in the presence of defects by constructing its B\"acklund transformation \cite{ara}.
Subsequently, Adans et al. considered the gauge Miura transformation for the entire supersymmetric KdV and mKdV hierarchies \cite{aaetal}.
Recently, nonlocal symmetries and B\"acklund transformation~(BT) of Kupershmidt's smKdV equation were constructed by Zhou, Tian and Li \cite{ztl}. Very Recently, for Hu's smKdV equation, Zhou, Tian and Yang \cite{zty} obtained its bi-Hamiltonian structure and Darboux transformations. Those authors also found the Darboux transformation~(DT) for Kupershmidt's smKdV equation via reduction.

 In this paper, we clarify the relationships
 among the super KdV/mKdV systems discussed above
 and demonstrate that their B\"acklund and Darboux transformations may be constructed from a unified viewpoint. We show that the three-component gsKdV system possesses two reductions, namely Kupershmidt's sKdV equation and Geng--Wu's sKdV equation. Then, we construct a few BDTs for gsKdV equation and display that through reductions one of them leads to BDT for Kupershmidt's sKdV equation and BDT for Geng--Wu equation. Furthermore, starting from a different spectral problem for gsKdV equation, we derive a Miura transformation and the corresponding modified system, namely gsmKdV equation. Consequently, several BDTs are constructed for gsmKdV equation and from them BDTs both for Kupershmidt's smKdV and Hu's smKdV equations are recovered. As applications, some solutions for sKdV/smKdV equations are calculated.\looseness=-1

 This paper is organized as follows.
 In Section \ref{sec2}, we introduce some basic knowledge including Grassmann algebra, differential and integral operators.
 In Section \ref{sec3}, we first review the gsKdV system, then construct a Miura transformation and the corresponding gsmKdV equation. We also derive the spectral problem for the gsmKdV equation based on that of the gsKdV equation. From the gsKdV and gsmKdV equations, we derive Kupershmit's sKdV/smKdV equations, Hu's smKdV equation, and Geng--Wu's sKdV equation. Additionally, we calculate the first negative flows for both the gsKdV and gsmKdV hierarchies.
 In Section \ref{sec4}, we present four BDTs for the gsKdV equation, and from the fourth one we recover a BDT for Kupershmidt's sKdV equation and obtain a BDT for Geng--Wu's sKdV equation.
 In Section \ref{sec5}, by considering the gauge matrix linear in the spectral parameter,
we construct five BTs of the gsmKdV equation, the first three are related to the four BTs of the gsKdV equation, the last one is used to derive the ones for Kupershmidt's smKdV and Geng--Wu's smKdV equations, respectively.
 In Section \ref{sec6}, the exact solutions of the gsKdV and gsmKdV equations are constructed. Section~\ref{sec7} summarizes the results.\looseness=-1

 \section{Preliminaries}\label{sec2}

In this section, we present some basic facts and properties of Grassmann algebras \cite{ber,yim}, differential and integral operators.

Let ${\mathcal G}$ be a ${\mathbb Z}_2$-graded algebra over a field of characteristics zero (such as ${\mathbb C}$ or ${\mathbb Q}$). Thus, ${\mathcal G}$ as a vector space is a direct sum ${\mathcal G}={\mathcal G}_0 \oplus {\mathcal G}_1$, such that ${\mathcal G}_i{\mathcal G}_j \subseteq {\mathcal G}_{i+j}$ (${\bmod} \,2$). Those elements of ${\mathcal G}$ that belong either to ${\mathcal G}_0$ or to ${\mathcal G}_1$ are called homogeneous, the ones in ${\mathcal G}_0$ are called even (bosonic, commuting), and those in ${\mathcal G}_1$ are called odd (fermionic or anticommuting).
In the article, we only consider homogeneous elements.

 Let $a$ be an element in $\cal G$, its parity $ |a|$ is defined to be $ 0 $ if $ a$ is even, and $ 1 $ if $ a$ is odd.
 The parity of the product $ab$ of two elements is a sum of their parities: $|ab|= |a|+|b|$.
Grassmann commutativity means that $ba=(-1)^{|a||b|}ab$ for any elements $a$ and $b$. In particular,
$ \xi \eta=-\eta \xi
$
 for all $\xi, \eta \in {\mathcal G}_1$ and even elements commute with all elements of ${\mathcal G}$. Moreover,
$\xi^n=\eta^n=0$, $ (\xi\eta)^n=0$,
for any integer $n\geq 2$. Thus any odd element is nilpotent.

Let $\bA$ and $\bD$ be $m\times m$ and $n\times n$ matrices with even entries, respectively.
Let $\bB$ and $\bC$ be~${m\times n}$ and $n\times m$ matrices with odd entries, respectively.
Then the super matrices considered in the present paper are of block type
$
 \bW=\bigl(\begin{smallmatrix}
 \bA& \bB\\
 \bC & \bD
 \end{smallmatrix}
 \bigr)$.
Such matrices in the standard format are even matrices \cite{yim}. Supertrace of $\bW$ is defined by
\[
\operatorname{str}({\bW})=\operatorname{tr}({\bA})-\operatorname{tr}({\bD}).\]
If
${\bA}$ or ${\bD}$ is invertible, then
the Berezinian (superdeterminant) of $\bW$ is defined as
\begin{equation*}
\operatorname{Ber}(\bW)=\det (\bA)\det \bigl(\bD-\bC\bA^{-1}\bB\bigr)^{-1}=\det \bigl(\bA-\bB\bD^{-1}\bC\bigr)\det (\bD)^{-1}.
\end{equation*}
If both the blocks
${\bA}$ and ${\bD}$ are invertible, then $\bW$ is invertible and its inverse matrix is given~by
\begin{gather*}
 \bW^{-1}=\left(\begin{matrix}
 \bigl(\bA-\bB\bD^{-1}\bC\bigr)^{-1} & -\bA^{-1}\bB\bigl(\bD-\bC\bA^{-1}\bB\bigr)^{-1}\\
 -\bD^{-1}\bC\bigl(\bA-\bB\bD^{-1}\bC\bigr)^{-1} & \bigl(\bD-\bC\bA^{-1}\bB\bigr)^{-1}
 \end{matrix}\right).
\end{gather*}

Let $\partial_x$ and $\partial^{-1}_x$ denote differentiation and integration with respect to an even variable $x$, respectively, satisfying identity $\partial_x\partial_x^{-1} =1$. They obey the following identities:
\begin{align*}
\partial_x^n \circ a = \sum_{k=0}^{n}\binom{n}{k}{\bigl(\partial_x^k a\bigr)} \partial_x^{n-k},\qquad
\partial_x^{-1}\circ a \partial_x^{n}=\sum_{k=0}^{n-1} (-1)^k {\bigl(\partial_x^k a\bigr)}\partial_x^{n-1-k}
+(-1)^{n} \partial_x^{-1}\circ (\partial_x^n a),
\end{align*}
where $a \in {\mathcal G}$, ${\partial_x^k a}=\frac{\partial^k a}{\partial x^k}$,
$n$ is a positive integer.
Therefore, the following formulae hold:
\begin{gather*}
\partial_x(\xi\eta)=(\xi \eta)_x= \xi_x \eta+\xi\eta_x,
\qquad
(\xi \xi_x)_x=\xi_x\xi_{x}+\xi\xi_{xx}=\xi\xi_{xx},
\\
\partial^{-1}_x (\xi_x) =\xi +\mu,\qquad
\partial^{-1}_x (\xi\eta_x) =\xi \eta-\partial^{-1}_x(\xi_x \eta ),\qquad
\partial^{-1}_x (\xi \xi_{xx}) =\xi \xi_x +c,
\end{gather*}
where $\xi, \eta \in {\mathcal G}_1$, $\mu$ and $c$ are odd and even integration constants, respectively.

\begin{Remark}
In the following sections, unless otherwise stated, Latin letters denote bosonic field variables, while Greek letters represent fermionic ones.
\end{Remark}

\section{Super KdV type hierarchies}\label{sec3}


In this section, we first employ a Miura transformation on the gsKdV equation to derive the gsmKdV equation. We then consider two reductions for both systems. Next, we analyze the flows within the gsKdV and gsmKdV hierarchies, in particular, we provide the first negative flows for both hierarchies.

\subsection{Super KdV type equations and Miura transformations}\label{sec3.1}
Let us start from the gsKdV equation of Holod and Pakuliak \cite{hp}
 \begin{subequations}\label{gskdv}
\begin{gather}
u_t=u_{xxx}-6u u_x+6 (\xi_{x}\eta- \xi\eta_{x})_x,\qquad
\eta_{t}=4\eta_{xxx}-6 u \eta_{x}-3 u_x \eta,\\
\xi_t= 4\xi_{xxx}-6 u \xi_{x}-3 u_{x}\xi,
\end{gather}
\end{subequations}
where subscripts denote partial derivatives, $t $ and $x$ are temporal and spatial variables, respectively; $ u= u(x,t)$ is bosonic
 field variable, $\xi=\xi(x,t)$, $ \eta=\eta(x,t)$ are fermionic field variables.
In addition to the space and time shift invariant, the system \eqref{gskdv} is also invariant under the following transformations:
\begin{gather}
u\rightarrow u-c,\qquad \partial_t\rightarrow \partial_t+6c \partial_x,\qquad \partial_x\rightarrow \partial_x;
\label{Ga}
\\
\xi\rightarrow c \xi, \qquad \eta\rightarrow c^{-1}\eta;\nonumber\\
u\rightarrow c^2 u,  \qquad \xi\rightarrow c^{\frac{3}{2}} \xi, \qquad \eta\rightarrow c^{\frac{3}{2}}\eta, \qquad\partial_x \rightarrow c \partial_x, \qquad \partial_t \rightarrow c^3 \partial_t,\nonumber
\end{gather}
where $c$ is a bosonic constant.

System \eqref{gskdv} has two different spectral problems, one is due to Holod and Pakuliak \cite{hp} and the other, given by Lou \cite{lou} recently, is
\begin{align}
\lambda \varphi=\varphi_{xx}-u\varphi-\frac{1}{2}\xi\partial^{-1} (\eta \varphi)+\frac{1}{2}\bigl(\partial^{-1}( \xi \varphi)\bigr) \eta,\qquad\varphi_ t =4\varphi_{xxx}-6u \varphi_x-3 u_x \varphi,\label{sp1}
\end{align}
where $\varphi\equiv\varphi(x,t;\lambda)$ is a wave function,
 $\lambda$ is a bosonic spectral parameter.
 Suppose that $\varphi$ is bosonic,
 let
$
 q\equiv \varphi_x/\varphi$, $ \zeta \equiv \bigl(\partial_x^{-1} (\eta \varphi)\bigr)/ \varphi$, $ \delta\equiv \bigl(\partial_x^{-1} (\xi \varphi)\bigr)/ \varphi$,
then \eqref{sp1} leads to
\begin{align}\label{G1}
u=q_{x}+q^2 -\frac{1}{2}(\delta_{x}\zeta-\delta \zeta_{x})-\lambda,\qquad
\eta=\zeta_{x}+q \zeta,\qquad
\xi=\delta_{x}+q \delta.
\end{align}
Inserting \eqref{G1} into \eqref{gskdv} gives rise to
\begin{gather*}
q_t=6\lambda q_x+\bigl(
q_{xx}-2 q^3+3q (\delta_{x}\zeta-\delta \zeta_{x})
\bigr)_x
+\frac{3}{2}(\delta_{x}\zeta-\delta \zeta_{x})_{xx},\\
\zeta_{t}= 6\lambda \zeta_x+ 4\zeta_{xxx}+6\bigl(q_{x}-q^2\bigr)\zeta_{x}+3\bigl(q_{x}-q^2\bigr)_x\zeta
+3 (\delta_{x}\zeta-\delta \zeta_{x}) \zeta_x -\frac{3}{2}(\delta_{x}\zeta-\delta \zeta_{x})_x \zeta,\\
\delta_{t}=6\lambda \delta_x+ 4\delta_{xxx}+6\bigl(q_{x}-q^2\bigr)\delta_{x}+3\bigl(q_{x}-q^2\bigr)_x\delta
+3 (\delta_{x}\zeta-\delta \zeta_{x}) \delta_x -\frac{3}{2}(\delta_{x}\zeta-\delta \zeta_{x})_x \delta.
\end{gather*}

 The argument presented thus far remains valid if $\xi$, $ \eta$, $ \zeta$, $ \delta$ are so-called Ren-fields \cite{lou}. However, in the rest of this paper, we consider only the case where they are fermionic fields.

It is noted that spectral problem \eqref{sp1} may be related to the one derived by Holod and Pakuliak. To see it,
by introducing $y=y(x,t;\lambda)$ such that
$y=q+\frac{1}{2}\delta\zeta,
$
 from \eqref{G1} one obtains
\begin{align}\label{m1}
u=y_{x}+y^2 -\xi \zeta-\lambda,\qquad
\eta=\zeta_{x}+y \zeta,\qquad
\xi=\delta_{x}+y \delta.
\end{align}
Furthermore, introduce a bosonic nonzero variable $\psi\equiv \psi(x,t;\lambda)$ such that
 $y={\psi_{x}}/{\psi}$. Then the last two equations of \eqref{m1} give
$
\zeta= \bigl( {\partial^{-1}_x(\eta\psi)}\bigr)/{ \psi}$, $
\delta =\bigl( {\partial^{-1}_x(\xi\psi)}\bigr)/{ \psi }$,
and the first equation of~\eqref{m1} leads to the spatial part of spectral problem \cite{hp}
\begin{align}\label{spec}
L \psi=\lambda \psi,\qquad L=\partial^2_x-u-\xi\partial_x^{-1}\circ \eta,
\end{align}
and the related temporal part is
\begin{equation}\label{psit}
\psi_t =4\bigl(L^{\frac{3}{2}}\bigr)_{+} \psi, \qquad 4\bigl(L^{\frac{3}{2}}\bigr)_{+}=4\partial_x^3-6u\partial_x-3u_x-6\xi\eta,
\end{equation}
where $(\,)_{+}$ refers to the differential part of a pseudo-differential operator.
Let$\bPsi=\bigl(\psi, \psi_x,\allowbreak\partial^{-1}(\eta\psi)\bigr)\smash{{}^{\mathsf{T}}}$, one may rewrite \eqref{spec} and \eqref{psit} as the matrix form
\begin{subequations}\label{skdv_sp}
\begin{align}\label{skdv_sp1}
&\bPsi_{x}={\bL}\bPsi,\qquad
{\bL}= \left(
\begin{matrix} 0& 1 & 0
\\
\lambda +u& 0 & \xi
\\
\eta & 0 & 0
\end{matrix}
\right),\\
&\bPsi_t={\bT} \bPsi,\qquad
{\bT}= \left(
\begin{matrix} u_x-2\xi\eta & 4\lambda-2u & 4\xi_x
\\
T_{21}& -u_x-2\xi\eta & 4\xi_{xx}+(4\lambda-2u)\xi
\\
4\eta_{xx}+(4\lambda-2u)\eta & -4\eta_x & -4\xi\eta
\end{matrix}
\right),\label{skdv_sp2}
\end{align}
\end{subequations}
where $T_{21}=u_{xx}+(\lambda+u)(4\lambda-2u)+2(\xi_x\eta-\xi\eta_x)$.

Notice that the above variables $y$, $ q$, $ \zeta$, $ \delta$ depend on $\lambda$. In fact,
under the Galilean boost \eqref{Ga} with $c=\lambda$, $\lambda$ can be eliminated.
Thus from \eqref{G1} we obtain the following Miura transformation:%
\begin{align}\label{miura2}
u=q_x+q^2 -\frac{1}{2}(\delta_{x}\zeta-\delta \zeta_{x}),\qquad
\eta=\zeta_{x}+q \zeta,\qquad
\xi=\delta_{x}+q \delta,
\end{align}
and the corresponding gsmKdV equation is given by
 \begin{subequations}\label{gsmkdv}
\begin{gather}
q_t=\bigl(
q_{xx}-2 q^3+3q (\delta_{x}\zeta-\delta \zeta_{x})
\bigr)_x
+\frac{3}{2}(\delta_{x}\zeta-\delta \zeta_{x})_{xx},\\
\zeta_{t}= 4\zeta_{xxx}+6\bigl(q_{x}-q^2\bigr)\zeta_{x}+3\bigl(q_{x}-q^2\bigr)_x\zeta
+3 (\delta_{x}\zeta-\delta \zeta_{x}) \zeta_x -\frac{3}{2}(\delta_{x}\zeta-\delta \zeta_{x})_x \zeta,\\
\delta_{t}= 4\delta_{xxx}+6\bigl(q_{x}-q^2\bigr)\delta_{x}+3\bigl(q_{x}-q^2\bigr)_x\delta
+3 (\delta_{x}\zeta-\delta \zeta_{x}) \delta_x -\frac{3}{2}(\delta_{x}\zeta-\delta \zeta_{x})_x \delta,
\end{gather}
\end{subequations}
 where
$
 q= q(x,t)$, $ \zeta= \zeta(x,t)$, $ \delta=\delta(x,t).
$
By
$
 y=q+\frac{1}{2}\delta\zeta,
$
\eqref{gsmkdv} may be equivalently formulated~as
\begin{gather*}
y_t = \bigl(y_{xx}-2 y^3\bigr)_x+3 (\zeta_{x}\delta_{x}-\zeta \delta_{xx}+y \zeta_{x} \delta-y \zeta \delta_x-y_x \zeta \delta)_x,\\
\zeta_{t}= 4\zeta_{xxx}+6\bigl(y_x-y^2\bigr)\zeta_{x}+3 \bigl(y_x-y^2\bigr)_{x}\zeta+3 \zeta_x(\delta_{x}+y \delta)\zeta,\\
\delta_t= 4\delta_{xxx}+6 \bigl(y_x-y^2 \bigr) \delta_{x}+3 \bigl(y_x-y^2\bigr)_{x}\delta+3 \zeta \delta_{xx} \delta -3(\zeta_x- y\zeta) \delta_x \delta.
\end{gather*}
It is interesting to point out that, through $\xi=\delta_{x}+y \delta$, above system may also be brought to Hu's smKdV equation \cite{hu}
 \begin{subequations}\label{smkdv1}
\begin{gather}
y_t= \bigl(y_{xx}-2 y^3+3 \zeta_{x}\xi-3\zeta \xi_x\bigr)_x,\\
\zeta_{t}= 4\zeta_{xxx}+6\bigl(y_x-y^2\bigr)\zeta_{x}+3 \bigl(y_x-y^2\bigr)_{x}\zeta+3 \zeta_x\xi\zeta,\\
\xi_t= 4\xi_{xxx}-6 \bigl(y_x+y^2 \bigr) \xi_{x}-3 \bigl(y_x+y^2\bigr)_{x}\xi-3\xi_x\zeta \xi.
\end{gather}
\end{subequations}

The Lax representation or spectral problem of the gsmKdV equation \eqref{gsmkdv} may be obtained from the spectral problem of gsKdV equation. Indeed,
consider the following gauge transformation:
\begin{align*}
\bPhi={\bG}\bPsi,
\qquad
{\bG}= \left(
\begin{matrix} 1 & 0 & 0
\\
-q+\frac{1}{2}\delta\zeta & 1 & -\delta
\\
-\zeta &0 &1
\end{matrix}
\right). 
\end{align*}
Then from \eqref{skdv_sp} and Miura transformation \eqref{miura2}, we have
\begin{subequations}\label{smkdv_sp}
\begin{align}
&\bPhi_{x}={\bU}\bPhi,\qquad
{\bU}= \left(
\begin{matrix} q+\frac{1}{2}\delta\zeta & 1 & \delta
\\
\lambda& -q+\frac{1}{2}\delta\zeta & 0
\\
0 &-\zeta & \delta\zeta
\end{matrix}
\right),\label{smkdv_sp1}\\
&\bPhi_{t}={\bM}\bPhi,
\qquad
{\bM}=( \bG_{t} +\bG \bT)\bG^{-1}.\label{smkdv_sp2}
\end{align}
\end{subequations}
 As $y=q+\frac{1}{2}\delta\zeta$, $ \zeta=\zeta$, $ \delta=\delta$ is an invertible change of variables, above $\bU$, $\bM$ may be equivalently rewritten in terms of $y$, $\zeta$ and $\delta$.

Now we consider the possible reductions of above three-component systems \eqref{gskdv} and \eqref{gsmkdv} together with the Miura transformation \eqref{miura2}, there are two cases as follows.

\begin{Reduction}
Consider the reduction condition $\zeta=\delta$ which implies $\xi=\eta$. Then from \eqref{gskdv} and \eqref{gsmkdv}, we have
Kupershmidt's sKdV equation
 \begin{gather}
 u_t=u_{xxx}-6uu_x+12\eta_{xx}\eta, \qquad
\eta_t=4\eta_{xxx}-6u\eta_x-3u_x\eta,\label{kuper-skdv}
 \end{gather}
and Kupershmidt's smKdV equation
\begin{subequations}\label{kuper-smkdv}
\begin{gather}
q_t=\bigl(
q_{xx}-2 q^3-6q \zeta \zeta_{x}
\bigr)_x
-{3}(\zeta\zeta_{x})_{xx},\\
\zeta_{t}= 4\zeta_{xxx}+6\bigl(q_{x}-q^2\bigr)\zeta_{x}+3\bigl(q_{x}-q^2\bigr)_x\zeta,
\end{gather}
\end{subequations}
 respectively.
The corresponding Miura transformation is given by
$
u=q_x+q^2 +\frac{1}{2}(\zeta_{x}\zeta-\zeta \zeta_{x} )$, $
\eta=\zeta_{x}+q \zeta$.
Thus, the results of Kupershmidt are recovered \cite{kuper1}.
\end{Reduction}

\begin{Reduction} Consider the condition $ \delta=
 \zeta_{xx}\bigl(1-\frac{1}{2} \zeta_x\zeta\bigr)+\bigl(q_x-q^2\bigr) \zeta
 $ that implies
 $\xi=\eta_{xx}-u \eta$.
Under this reduction, \eqref{gskdv} is reduced to
\begin{subequations}\label{geng0}
\begin{align}
u_t=&u_{xxx}-6u u_x+6(\eta_x\eta_{xx}-\eta\eta_{xxx}+2 u\eta\eta_{x})_x,\\
\eta_t=& 4\eta_{xxx}-6 u \eta_{x}-3 u_{x}\eta,
\end{align}
\end{subequations}
which is equivalent to Geng--Wu's sKdV equation \cite{gw}
\begin{subequations}\label{geng}
\begin{gather}
v_t=v_{xxx}-6v v_x-12 v \eta_{xx}\eta-6 v_x \eta_x \eta+3 \eta_{xxxx}\eta+6 \eta_{xxx} \eta_{x},\\
\eta_t= 4\eta_{xxx}-6 v \eta_{x}-3 v_{x}\eta,
\end{gather}
\end{subequations}
by $v=u+\eta\eta_x$. And the modified equation \eqref{gsmkdv} becomes
\begin{subequations}\label{m-geng}
\begin{gather}
q_t=\bigl(
q_{xx}-2 q^3+3q a
\bigr)_x
+\frac{3}{2}a_{xx},\\
\zeta_{t}= 4\zeta_{xxx}+6\bigl(q_{x}-q^2\bigr)\zeta_{x}+3\bigl(q_{x}-q^2\bigr)_x\zeta
+3 \zeta_{xxx}\zeta \zeta_x,
\end{gather}
\end{subequations}
with $a\equiv \zeta_{xxx}\zeta-\zeta_{xx}\zeta_x +2\bigl(q_x-q^2\bigr)\zeta_x \zeta$.
The corresponding Miura transformation is given by~${
v=q_x+q^2-q\zeta_{xx}\zeta -\frac{1}{2}(\zeta_{xx}\zeta)_x}$, $ \eta=\zeta_{x}+q \zeta$.
\end{Reduction}

\begin{Remark}
It is remarked that the super KdV and mKdV equations described above are fermionic extensions of their classical counterparts, rather than supersymmetric ones.
Taking Kupershmidt's sKdV equation \eqref{kuper-skdv} as an example, we now compare it with the supersymmetric KdV equation. The well-known supersymmetric KdV equation in component form \cite{MR,mathieu} is given by
\begin{gather}
u_t=u_{xxx}+6uu_x-3\eta\eta_{xx},\qquad
\eta_t=\eta_{xxx}+3(u\eta)_x.\label{susykdv1}
\end{gather}
As observed by Mathieu \cite{mathieu}, while \eqref{kuper-skdv} and \eqref{susykdv1} appear quite similar, they are fundamentally different. In fact,
 \eqref{susykdv1}
 is invariant under the following supersymmetric transformation
$
 u \rightarrow u+\mu\eta_x$, $ \eta\rightarrow \eta+\mu u$,
where $\mu$ is a fermionic parameter. In this situation, by extending the spatial variable $x$ to a~doublet $(x, \theta)$ where $\theta$ is an odd variable, one may introduce a super space derivative
${\cal D}=\partial_\theta+\theta\partial_x$, and a fermionic super field $\alpha=\alpha(t,x,\theta)$ such that
$
\alpha=\eta(x,t)+\theta u(x,t)$.
In this way~\eqref{susykdv1} is rewritten as the following form:
 \begin{equation}\label{susykdv}
\alpha_t=\alpha_{xxx}+3(\alpha ({\cal D}\alpha))_x.
\end{equation}
 However, \eqref{kuper-skdv} is not invariant under either of the following supersymmetric transformations~${
\eta \rightarrow \eta +\frac{\mu u}{c}}$, $ u \rightarrow u + c\mu \eta_x $,
or
$u \rightarrow u +\frac{ \mu \eta}{c}$, $ \eta \rightarrow \eta + c\mu u_x$,
where
$c$ is a certain bosonic constant.
The reason is that the coefficients of $u_t$, $u_{xxx}$ and $\eta_t$, $\eta_{xxx}$ in \eqref{kuper-skdv} or \eqref{geng0} are disproportionate thus it is impossible to combine these equations.
 Thus~\eqref{susykdv1} or \eqref{susykdv} is known as the supersymmetric KdV equation while 
 \eqref{kuper-skdv} and \eqref{geng0} are the fermionic or super KdV equations.
 \end{Remark}

\subsection{The flows of the gsKdV and gsmKdV hierarchies}

 Let us first consider the positive flows of the gsKdV hierarchy. In fact, the Lax pairs are given~by%
\begin{equation}
L \psi=\lambda \psi,\qquad
\psi_{t_N} =4\bigl(L^{\frac{N}{2}}\bigr)_{+}\psi,\label{lax-N}
\end{equation}
where $
 u= u(x,t_N)$, $ \xi= \xi(x,t_N)$, $ \eta=\eta(x,t_N)$, $ N=1, 3, 5,\dots, t_3=t$,
then Lax equation
\smash{$
L_{t_N}=\big[4\bigl(L^{\frac{N}{2}}\bigr)_{+}, L\big]
$}
provides $N$-th flow of the gsKdV hierarchy. One may also write \eqref{lax-N} as the matrix form
\begin{equation}\label{PsixtN}
\bPsi_x={\bL} \bPsi,\qquad \bPsi_{t_N}={\bT_{N}} \bPsi,
\end{equation}
where $\bT_{3}=\bT$, and the $N$-th flow of the gsKdV hierarchy can be derived by the zero curvature condition
\begin{equation}\label{zero}
\bL_{t_N}-\bT_{N,x}+[\bL,\bT_N]=0.
\end{equation}

From \eqref{PsixtN} and Miura transformation \eqref{miura2}, we have
\begin{gather}
\bPhi_{x}={\bU}\bPhi,\qquad
\bPhi_{t_N}={\bM}_N\bPhi,
\qquad
{\bM}_N=( \bG_{t_N} +\bG \bT_N)\bG^{-1},\label{smkdv_spN}
\end{gather}
where
$
 q= q(x,t_N)$, $ \zeta= \zeta(x,t_N)$, $ \delta= \delta(x,t_N)$, $ \bM_{3}=\bM$.
 Then the zero curvature condition
\begin{equation}\label{zero2}
\bU_{t_N}-\bM_{N,x}+[\bU,\bM_N]=0
\end{equation}
gives rise to $N$-th flow of the gsmKdV hierarchy.

It is possible to construct the negative flows. Let $N$ be a negative integer and take $\bT_N$, $ \bM_{N}$ as follows:
\begin{gather*}
\bT_{N} =\lambda^{N}\bT^{(N)}+\lambda^{N+1}\bT^{(N+1)}+\cdots +\bT^{(0)},\\
\bM_{N}=\lambda^{N}\bM^{(N)}+\lambda^{N+1}\bM^{(N+1)}+\cdots+ \bM^{(0)},
\end{gather*}
where \smash{$\bT^{(i)}$}, \smash{$\bM^{(i)}$}, $i=N,\dots,0$, are independent in $\lambda$. By choosing suitable \smash{$\bT^{(i)}$}, \smash{$\bM^{(i)}$},
 the zero curvature conditions \eqref{zero} and \eqref{zero2} yield
 the negative flows of the gsKdV and gsmKdV hierarchies, respectively. We provide one example below.

Let
\begin{align*}
 {\bT}_{-1}= \frac{1}{\lambda}\left(
\begin{matrix} \frac{Z-w_{x,t_{-1}}}{2} & w_{t_{-1}} & -\xi_{t_{-1}}
\\
\lambda w_{t_{-1}}+2w_xw_{t_{-1}}-\frac{1}{2}w_{xx,t_{-1}}+\frac{1}{2}(\xi\eta_{t_{-1}}-\xi_{t_{-1}}\eta) & \frac{Z+w_{x,t_{-1}}}{2} & w_{t_{-1}}\xi-\xi_{x,t_{-1}}
\\
\eta w_{t_{-1}}-\eta_{x,t_{-1}} & \eta_{t_{-1}} & Z
\end{matrix}
\right),
\end{align*}
where $u=2w_x$, $Z$ is determined by $Z_x=(\xi\eta)_{t_{-1}}$, then \eqref{zero} leads to{\samepage
\begin{subequations}\label{gskdv7}
\begin{gather}
w_{xxx,t_{-1}}=8w_xw_{x,t_{-1}}+4w_{xx}w_{t_{-1}}+3(\xi\eta_{x,t_{-1}}- \xi_{x,t_{-1}}\eta)+\xi_{x}\eta_{t_{-1}}-\xi_{t_{-1}}\eta_{x},\\
\xi_{xx,t_{-1}} =\frac{3}{2}w_{x,t_{-1}}\xi+2\xi_{t_{-1}}w_x+\xi_xw_{t_{-1}}-\frac{1}{2}\xi Z,\\
\eta_{xx,t_{-1}} =\frac{3}{2}w_{x,t_{-1}}\eta+2\eta_{t_{-1}} w_x+\eta_xw_{t_{-1}}+\frac{1}{2}\eta Z,
\end{gather}
\end{subequations}
which is the first negative flow of the gsKdV hierarchy.}

 Let
\begin{align*}
 {\bM}_{-1}= \left(
\begin{matrix}
0 & \frac{{\rm e}^{2r}}{\lambda} &0
\\
{\rm e}^{-2r} & 0 & -\delta_{t_{-1}}
\\
-\zeta_{t_{-1}} & 0 & 0
\end{matrix}
\right),
\end{align*}
where $q=r_x$, then \eqref{zero2} leads to
\begin{subequations}\label{sinh}
\begin{gather}
r_{x,t_{-1}}=2\sinh{2r}-\frac{1}{2}(\delta_{t_{-1}} \zeta+\zeta_{t_{-1}}\delta),\qquad
\delta_{x,t_{-1}} ={\rm e}^{-2r}\delta -r_x \delta_{t_{-1}} -\frac{1}{2} \delta\zeta \delta_{t_{-1}},\\
\zeta_{x,t_{-1}} ={\rm e}^{-2r}\zeta -r_x \zeta_{t_{-1}} -\frac{1}{2} \zeta\delta \zeta_{t_{-1}},
\end{gather}
\end{subequations}
which is the first negative flow of the gsmKdV hierarchy, i.e., a super extension of the sinh-Gordon equation.

Upon substituting the Miura transformation \eqref{miura2} and the equation \eqref{sinh} into \eqref{gskdv7}, we have
$
w_{t_{-1}}={\rm e}^{2r}$, $ Z= {\rm e}^{2r} \delta \zeta$.
The first equation is the temporal Miura transformation which is exactly the one appeared in~\cite{yfa}.

We may also consider the two reductions mentioned in Section \ref{sec3.1} and obtain the negative flows of Kupershmidt's sKdV hierarchy and Geng--Wu's sKdV hierarchy.

\begin{Reduction}
Consider the reduction condition $\zeta=\delta$ which implies $\xi=\eta$, $ Z=0$. Then from \eqref{gskdv7} and \eqref{sinh} we have the first negative flow of
Kupershmidt's sKdV hierarchy
 \begin{gather*}
w_{xxx,t_{-1}}=8w_xw_{x,t_{-1}}+4w_{xx}w_{t_{-1}}+6\eta\eta_{x,t_{-1}}+2\eta_{x}\eta_{t_{-1}},\\
\eta_{xx,t_{-1}} =\frac{3}{2}w_{x,t_{-1}}\eta+2\eta_{t_{-1}}w_x+\eta_xw_{t_{-1}},
 \end{gather*}
and the first negative flow of Kupershmidt's smKdV hierarchy
\begin{gather*}
r_{x,t_{-1}}=2\sinh{2r}-\zeta_{t_{-1}} \zeta,\qquad
\zeta_{x,t_{-1}} ={\rm e}^{-2r}\zeta -r_x \zeta_{t_{-1}},
\end{gather*}
 respectively. The corresponding Miura transformation is given by
\[
w_x=\frac{1}{2}\big(r_{xx}+r_x^2\big)-\frac{1}{2}\zeta_{x}\zeta,
 \qquad \eta=\zeta_{x}+r_x \zeta, \qquad w_{t_{-1}}={\rm e}^{2r}.
\]
\end{Reduction}

\begin{Reduction} Consider the condition $ \delta=
 \zeta_{xx}-\frac{1}{2}
 \zeta_{xx} \zeta_x\zeta+\bigl(r_{xx}-r_x^2\bigr) \zeta$ which gives
 $\xi=\eta_{xx}-2w_x \eta$, $ Z= (\eta_x\eta)_{t_{-1}}={\rm e}^{2r}\zeta_{xx}\zeta$.
Under this reduction, \eqref{gskdv7} is reduced to
\begin{gather*}
w_{xxx,t_{-1}}=8w_xw_{x,t_{-1}}+4w_{xx}w_{t_{-1}}+(3\eta\partial_{x}\partial_{t_{-1}}-3\eta_{x,t_{-1}}-\eta_{t_{-1}}\partial_{x}+\eta_{x}\partial_{t_{-1}}
)(\eta_{xx}-2w_x \eta),\\
\eta_{xx,t_{-1}} =\frac{3}{2}w_{x,t_{-1}}\eta+2\eta_{t_{-1}} w_x+\eta_xw_{t_{-1}}+\frac{1}{2}\eta \eta_{x}\eta_{t_{-1}},
\end{gather*}
whose modified equation becomes
\begin{gather*}
r_{x,t_{-1}}=2\sinh{2r}-\frac{1}{4}(\zeta_{t_{-1}}- \zeta\partial_{t_{-1}})\bigl(2\zeta_{xx}- \zeta_{xx}\zeta_x\zeta+2\bigl(r_{xx}-r_x^2\bigr) \zeta\bigr),\\
\zeta_{x,t_{-1}} ={\rm e}^{-2r}\zeta -r_x \zeta_{t_{-1}}-\frac{1}{2} \zeta\zeta_{xx} \zeta_{t_{-1}}.
\end{gather*}
The corresponding Miura transformation is given by
\[
w_x=\frac{1}{2}\bigl(r_{xx}+r_x^2\bigr)+\frac{1}{4}(\zeta_{xx}\zeta_x-\zeta_{xxx}\zeta)-\frac{1}{2}\zeta_x\zeta \bigl(r_{xx}-r_x^2\bigr),
 \qquad \eta=\zeta_{x}+r_x \zeta, \qquad w_{t_{-1}}={\rm e}^{2r}.
\]
\end{Reduction}

\section{B\"acklund--Darboux transformations for gsKdV equation}\label{sec4}
In this section, for gsKdV equation \eqref{gskdv} we first recall BDTs \cite{xl} and build new BDTs. Then we will show that a BDT for Geng--Wu's sKdV equation \eqref{geng} may be resulted by a proper reduction.
 For convenience, we introduce the potential variables defined by
\begin{equation*}
u=2w_x, \qquad \widetilde{u}=2\widetilde{w}_x,\qquad\widehat{u}=2\widehat{w}_x,\qquad \widehat{\widetilde{u}}=2\widehat{\widetilde{w}}_x,\qquad \overline{u}=2\overline{w}_x.
\end{equation*}



Let $\bPsi$ be the solution of \eqref{PsixtN}.
Consider a gauge transformation
\begin{equation}\label{W0}
\widetilde{\bPsi}=\bW\bPsi,\qquad
\end{equation}
such that
\begin{equation}\label{Psixt}
\widetilde{\bPsi}_x=\widetilde{\bL}\widetilde{\bPsi},\qquad \widetilde{\bPsi}_{t_N}=\widetilde{\bT}_N\widetilde{\bPsi},
\end{equation}
where $\widetilde{\bL}$ and $\widetilde{\bT}_N$ are $\bL$ and $\bT_N$ but with $u$, $ \xi$ and $\eta$ replaced by $\widetilde{u}$, $ \widetilde{\xi}$ and $\widetilde{\eta}$, respectively.
The compatibility of \eqref{PsixtN}, \eqref{W0} and \eqref{Psixt} lead to the conditions $\widetilde{\bPsi}_x=\bigl(\widetilde{\bPsi}\bigr)_x$ and $\widetilde{\bPsi}_{t_N}=\bigl(\widetilde{\bPsi}\bigr)_{t_N}$, which yield
\begin{subequations}\label{Wxt}
\begin{align}
\label{Wx}
&\bW_x+\bW \bL- \widetilde{\bL}\bW=0,\\
 &\bW_{t_N}+\bW \bT_N- \widetilde{\bT}_N\bW=0,\label{Wt}
\end{align}
\end{subequations}
respectively.
The gauge transformation $\widetilde{\bPsi} = \bW\bPsi$, together with condition~\eqref{Wx},
 defines a DBT for the spectral problem $\bPsi_x=\bL\Psi$,
 while the transformation, combined with~\eqref{Wxt}, yields a DBT for~\eqref{PsixtN}, i.e., for $t_N$-flow of the gsKdV hierarchy. Since all flows share the identical spatial part of the associated linear problem, the entire hierarchy possesses the congruent Darboux matrix and spatial part of the BT, and from \eqref{Wt} one may obtain the temporal part of BT.
 Due to the associated Lax matrices are traceless, hence the derived Darboux matrices have constant determinants. In the following we mainly consider the $N=3$ case, namely the third-order flow.

Suppose that $\bW$ is linear in $\lambda$. Then we have the following four DBTs for \eqref{skdv_sp} and the related systems.

\begin{Proposition}\label{prp1}Let
\begin{gather}
\widetilde{\bPsi}={\bW}_1\bPsi,\nonumber
\\
{\bW}_1=\bW_1(\lambda_1, w,\xi,\widetilde{w},\widetilde{\eta})= \left(
\begin{matrix}
\widetilde{w}-w & 1 &0
\\
 \lambda- \lambda_{1}+{(\widetilde{w}-w)}^2+\xi\widetilde{\eta} & \widetilde{w}-w & {\xi}
\\
\widetilde{\eta}& 0 & 1
\end{matrix}
\right).\label{dt1}
\end{gather}
Then \eqref{skdv_sp1}, \eqref{dt1} and $\widetilde{\bPsi}_x=\widetilde{\bL}\widetilde{\bPsi}$ are compatible if and only if
the following BT
\begin{align}\label{bt1}
\widetilde{\eta}_x=-\eta+(\widetilde{w}-w)\widetilde{\eta},
\qquad
\xi_x=\widetilde{\xi}-(\widetilde{w}-w)\xi,\qquad
\widetilde{w}_x=-w_{x}+(\widetilde{w}-w)^2- \lambda_1+\xi\widetilde{\eta}
\end{align}
holds. This BT will be referred to as ${\rm BT}_{\rm sKdV}\bigl(\lambda_1; w,\xi,\eta,\widetilde{w},\widetilde{\xi},\widetilde{\eta}\bigr)$.
\end{Proposition}
\begin{Proposition}\label{prop2}
Let
\begin{align}
\label{dt2}
\widehat{\bPsi}={\bW}_2\bPsi,
\qquad
{\bW}_2=
 \left(
\begin{matrix}
\widehat{w}-w & 1 & -\widehat{\xi}
\\
 \lambda- \lambda_2+{(\widehat{w}-w)}^2 & \widehat{w}-w & \bigl(w-\widehat{w}\bigr) \widehat{\xi}
\\
\bigl(w-\widehat{w}\bigr)\eta & -\eta & \lambda- \lambda_2 -\widehat{\xi}\eta
\end{matrix}
\right).
\end{align}
Then \eqref{skdv_sp1}, \eqref {dt2} and $\widehat{\bPsi}_x=\widehat{\bL}\widehat{\bPsi}$ are consistent if and only if
the following BT
\begin{align}\label{bt2}
{\eta}_x=-\widehat{\eta}-\bigl(\widehat{w}-w\bigr){\eta},\qquad
\widehat{\xi}_x=\xi+\bigl(\widehat{w}-w\bigr)\widehat{\xi},\qquad
\widehat{w}_x=-w_{x}+{\bigl(\widehat{w}-w\bigr)}^2- \lambda_2+\widehat{\xi}{\eta}
\end{align}
holds. This BT will be denoted by \smash{${\rm BT}_{\rm sKdV}\bigl(\lambda_2;\widehat{w},\widehat{\xi},\widehat{\eta}, w,\xi,\eta\bigr)$}.
\end{Proposition}

The BDTs presented above are elementary BDTs. Now we use them to construct a compound BDT. Let us
consider \smash{${\rm BT}_{\rm sKdV}\bigl(\lambda_1; \widehat{w}, \widehat{\xi}, \widehat{\eta}, \widehat{\widetilde{w}}, \widehat{\widetilde{\xi}}, \widehat{\widetilde{\eta}}\bigr)$}
 and \smash{${\rm BT}_{\rm sKdV}\bigl(\lambda_2;\widehat{w},\widehat{\xi},\widehat{\eta}, w,\xi,\eta\bigr)$}. Then, eliminating $\widehat{\eta}$, $ \widehat{\xi}_x$, $\widehat{w}_x$ in those BTs yields
 \begin{gather}
 \bigl(\widehat{\widetilde{\eta}}-\eta\bigr)_x =\bigl(\widehat{\widetilde{w}}-w\bigr)\widehat{\widetilde{\eta}}-\bigl(\widehat{w}-w\bigr) \bigl(\widehat{\widetilde{\eta}}-\eta\bigr), \label{eta12x}
\\
\widehat{\widetilde{\xi}}-\xi =\bigl(\widehat{\widetilde{w}}-w\bigr)\widehat{\xi}, \label{xi12x}
\\
\bigl( \widehat{\widetilde{w}}-w\bigr)_{x}=\lambda_2-\lambda_1+\bigl(\widehat{\widetilde{w}}-w\bigr)^2-2\bigl(\widehat{\widetilde{w}}-w\bigr) \bigl(\widehat{w}-w\bigr)+\widehat{\xi}\bigl(\widehat{\widetilde{\eta}}-{\eta}\bigr).\label{w12x}
 \end{gather}
Equations \eqref{xi12x} and \eqref{w12x} allow us to obtain
 \begin{equation}\label{xi12w12}
 \widehat{\xi}=\frac{\widehat{\widetilde{\xi}}-\xi}{\widehat{\widetilde{w}}-w},
\qquad
 \widehat{w}=w+\frac{\widehat{\widetilde{w}}-w}{2}+\frac{\lambda_2-\lambda_1-\bigl( \widehat{\widetilde{w}}-w\bigr)_{x}+\widehat{\xi}\bigl(\widehat{\widetilde{\eta}}-{\eta}\bigr)}{2\bigl(\widehat{\widetilde{w}}-w\bigr)}.
 \end{equation}
 Inserting them into \eqref{eta12x} and the last two equations of \eqref{bt2}, one finds
 a BT with two parameters, that is
 \begin{subequations}\label{bt3}
\begin{gather}
\bigl(\widehat{\widetilde{\eta}}-\eta\bigr)_x =
	\frac{1}{2}\bigl(\widehat{\widetilde{w}}-w\bigr)\bigl(\widehat{\widetilde{\eta}}+\eta\bigr)
	+\frac{\bigl(\lambda_1-\lambda_2+\widehat{\widetilde{w}}_x-w_x\bigr)\bigl(\widehat{\widetilde{\eta}}-\eta\bigr)}{2\bigl(\widehat{\widetilde{w}}-w\bigr)},\\
	\bigl(\widehat{\widetilde{\xi}}-\xi\bigr)_x =
	\frac{1}{2}\bigl(\widehat{\widetilde{w}}-w\bigr)\bigl(\widehat{\widetilde{\xi}}+\xi\bigr)
	+\frac{\bigl(\lambda_2-\lambda_1+\widehat{\widetilde{w}}_x-w_x\bigr)\bigl(\widehat{\widetilde{\xi}}-\xi\bigr)}{2\bigl(\widehat{\widetilde{w}}-w\bigr)},
	\\
	\bigl(\widehat{\widetilde{w}}-w\bigr)_{xx}= \bigl(\lambda_2+\lambda_1+2\widehat{\widetilde{w}}_x+2w_x\bigr)\bigl(\widehat{\widetilde{w}}-w\bigr)-\frac{1}{2}\bigl(\widehat{\widetilde{w}}-w\bigr)^3
+\frac{\bigl(\widehat{\widetilde{w}}-w\bigr)^2_{x}-(\lambda_2-\lambda_1)^2}{2\bigl(\widehat{\widetilde{w}}-w\bigr)}
	\nonumber\\
	\phantom{\bigl(\widehat{\widetilde{w}}-w\bigr)_{xx}= }{} -\frac{(\lambda_2-\lambda_1)\bigl(\widehat{\widetilde{\xi}}-\xi\bigr)\bigl(\widehat{\widetilde{\eta}}-\eta\bigr)}{\bigl(\widehat{\widetilde{w}}-w\bigr)^2}
	+\xi\widehat{\widetilde{\eta}}- \widehat{\widetilde{\xi}}\eta.
\end{gather}
\end{subequations}
The composition of these two DTs \eqref{dt1} and \eqref{dt2} leads to
\smash{$
\widehat{\widetilde{\bPsi}}
=\widehat{ \bW}_1 \widehat{\bPsi}=\widehat{ \bW}_1\bW_2\bPsi$}.
Thus we have the following proposition.

\begin{Proposition}\label{prop3}
Let
\begin{align}\label{dt3}
	\widehat{\widetilde{\bPsi}}={\bW_3}\bPsi,
\qquad
{\bW_3}
=&
\left(
\begin{matrix}
\lambda-\lambda_2+\bigl(\widehat{w}-w\bigr)\bigl(\widehat{\widetilde{w}}-{w}\bigr)
&\widehat{\widetilde{w}}-w & \bigl(w-\widehat{\widetilde{w}}\bigr)\widehat{\xi}\\
W_{21} & W_{22} & W_{23}
\\
\bigl(\widehat{w}-w\bigr)\bigl(\widehat{\widetilde{\eta}}-\eta\bigr) & \widehat{\widetilde{\eta}}-\eta
& \lambda- \lambda_2 +\widehat{\xi}\bigl(\widehat{\widetilde{\eta}}-\eta\bigr)
\end{matrix}
\right),
\end{align}
where
\begin{gather*}
	W_{21}=\big[\lambda-\lambda_{1}+\bigl(\widehat{\widetilde{w}}-\widehat{w}\bigr)^2+\widehat{\xi} \bigl(\widehat{\widetilde{\eta}}-\eta\bigr)\big]\bigl(\widehat{w}-w\bigr)+\big[\lambda-\lambda_{2}+\bigl(\widehat{w}-w\bigr)^2\big]\bigl(\widehat{\widetilde{w}}-\widehat{w}\bigr),\\
	W_{22}=\lambda-\lambda_{1}+\bigl(\widehat{\widetilde{w}}-w\bigr)\bigl(\widehat{\widetilde{w}}-\widehat{w}\bigr)+\widehat{\xi}\bigl( \widehat{\widetilde{\eta}}-\eta\bigr),\qquad
	W_{23}=-\big[\lambda_{2}-\lambda_{1}+\bigl(\widehat{\widetilde{w}}-w\bigr)\bigl(\widehat{\widetilde{w}}-\widehat{w}\bigr)\big]\widehat{\xi},
\end{gather*}
with the intermediate variables $\widehat{w}$ and $\widehat{\xi}$ are defined by \eqref{xi12w12}.
Then \eqref{skdv_sp1}, \eqref {dt3} and \smash{$\widehat{\widetilde\bPsi}_x=\widehat{\widetilde\bL}\widehat{\widetilde\bPsi}$} are consistent if and only if \eqref{bt3} is satisfied.
\end{Proposition}

Those BDTs may be reformulated in terms of solutions of the spectral problem \eqref{skdv_sp} and its adjoint problem
\begin{equation}\label{skdv_sp3}
\bPsi_{x}^+=-\bPsi^+{\bL},\qquad \bPsi_{t}^+=-\bPsi^+{\bT},
\end{equation}
where
$\bPsi^+=\bigl(- \psi^+_{x},\psi^+, -\partial^{-1}_x\bigl(\psi^+\xi\bigr)\bigr)$.
To do so, let \smash{$\bPsi_j=\bigl(\psi_j, \psi_{j,x}, \partial_x^{-1}(\eta\psi_j)\bigr)^{\mathsf{T}}$} be a solution of \eqref{skdv_sp} at $\lambda=\lambda_j$ and $\bPsi^+_j=\bigl(- \psi^+_{j,x},\psi^+_j, -\partial^{-1}_x\bigl(\psi_j^+\xi\bigr)\bigr)$ be a solution of \eqref{skdv_sp3}
 at $\lambda=\lambda_j$, where~\smash{${\psi_j=\psi_j(x,t;\lambda_j)}$}, \smash{$ \psi_j^+= \psi_j^+(x,t;\lambda_j)$} are bosonic, for $j=1, 2$. Introduce new variables $y_j$, $\zeta_j$, \smash{$y^+_j$}, \smash{$\delta_j^+$} such that
\begin{equation}
y_j\equiv \frac{\psi_{j,x}}{\psi_j},\qquad
\zeta_j\equiv \frac{\partial_x^{-1}(\eta\psi_j)}{ \psi_j},\qquad
y_j^+\equiv \frac{\psi^+_{j,x}}{\psi_j^+},\qquad
\delta_j^+\equiv\frac{\partial_x^{-1}\bigl(\xi\psi_j^+\bigr)}{ \psi_{j}^+ }.\label{transj}
\end{equation}
Then, it is easy to see that those variables satisfy
\begin{subequations}\label{y1x}
\begin{align}
& y_{j,x}=u-y_j^2 +\lambda_j+ \xi \zeta_j,\qquad \zeta_{j,x}=\eta-y_j \zeta_j,\\
&y^+_{j,x}=u-\bigl(y_j^+\bigr)^2 +\lambda_j+\eta\delta^+_j , \qquad
\delta_{j,x}^+=\xi- y_j^+ \delta_j^+.
\end{align}
\end{subequations}
The above equations imply that we may take
\begin{equation}\label{yj+}
y_j^+=y_j-\delta_j^+ \zeta_j.
\end{equation}

Using the variables defined by \eqref{transj}, we may rewrite the above three DBTs and obtain the following proposition.
\begin{Proposition} \label{prop4}
 Let
\begin{align}
&\widetilde{w}=w-y_1,
\qquad
\widetilde{\eta} = -\zeta_1,\qquad
 \widetilde{\xi}= \xi_x- {y_1\xi};
 \label{w1}\\
&
\widehat{w}=w -y_2^+,
\qquad \widehat{\eta}= -\eta_x+y_2^+\eta,
\qquad
\widehat{\xi}=\delta_2^+;
\label{w2}\\ &
\widehat{\widetilde{w}}={w}+
G,
\qquad
\widehat{\widetilde{\eta}}={\eta}+
{G\zeta_1 },
\qquad
\widehat{\widetilde{\xi}}=\xi+
{G\delta_2^+ },
\label{w3}
\end{align}
where
\[
G\equiv\frac{\lambda_1-\lambda_2}{ y_2^+-y_1+\delta_2^+\zeta_1},
\]
 $y_1$, $\zeta_1$, $y_2^+$, $\delta^+_2$ are defined by \eqref{transj}.
Then, for any $(\bPsi, u,\xi,\eta) $ satisfying equation~\eqref{skdv_sp}, \smash{$\bigl(\widetilde{\bPsi}, \widetilde{w}, \widetilde{\xi}, \widetilde{\eta}\bigr)$} given by \eqref{dt1}, \eqref{w1}, \smash{$\bigl(\widehat{\bPsi}, \widehat{w}, \widehat{\xi}, \widehat{\eta}\bigr)$} given by \eqref{dt2}, \eqref{w2}, and \smash{$\bigl(\widehat{\widetilde{\bPsi}}, \widehat{\widetilde{w}}, \widehat{\widetilde{\xi}}, \widehat{\widetilde{\eta}}\bigr)$} given by \eqref{dt3}, \eqref{w3} satisfy the following equations:
\[
\widetilde{\bPsi}_x=\widetilde{\bL}\widetilde{\bPsi}, \qquad
\widehat{\bPsi}_x=\widehat{\bL}\widehat{\bPsi}, \qquad
\widehat{\widetilde\bPsi}_x=\widehat{\widetilde\bL}\widehat{\widetilde\bPsi}, \qquad \widehat{\widetilde\bPsi}_t=\widehat{\widetilde\bT}\widehat{\widetilde\bPsi},
\]
 respectively.
\end{Proposition}

\begin{Remark}
Inserting \eqref{w3} into \eqref{bt3} yields
\begin{gather*}
\eta=\zeta_{1,x}+ \frac{{\bigl(G_x-G^2+\lambda_2-\lambda_1\bigr)\zeta_1 }}{2G},\qquad
\xi=\delta_{2,x}^+ + \frac{{\bigl(G_x-G^2+\lambda_1-\lambda_2\bigr)\delta_2^+}}{2G},\\
u=\frac{G_{xx}}{2G}-{ G_x}+\frac{G^2}{4}+\frac{(\lambda_2-\lambda_1)^2-G_x^2}{4G^2}-\frac{(\lambda_1+\lambda_2)}{2}\\
\phantom{u=}{}+\frac{1}{2}\bigl(\delta_2^+ \zeta_{1,x}-\delta_{2,x}^+\zeta_1\bigr)+\frac{(\lambda_2-\lambda_1)\delta_2^+\zeta_1} {G},
\end{gather*}
which can be viewed as an alternative form for the transformation of $u$, $\xi$, $\eta$.
\end{Remark}

Next, we consider the limit of the compound BDT and show that a fourth BDT may be obtained. Let $\lambda_2=\lambda_1+\epsilon$, $\psi_{2}=\psi_1(x,t;\lambda_1+\epsilon)$. Taking account of \eqref{yj+}, the limits of \eqref{bt3}, \eqref{dt3} and \eqref{w3} as $\epsilon \rightarrow 0$ give rise to the following proposition.

\begin{Proposition} \label{prop5}
 Let
\begin{align}
\label{w4}
& {\overline{w}} =w-\frac{1}{y_{1}'-\delta_{1}^+\zeta_{1}'},
\qquad {\overline{\eta}}=\eta-\frac{\zeta_1}{y_{1}'-\delta_{1}^+ \zeta_{1}'},
\qquad {\overline{\xi}}=\xi-\frac{\delta_1^+}{y_{1}'-\delta_{1}^+\zeta_{1}'},
\\\label{dt4}
&	 {\overline{\bPsi}}={\bW_4}\bPsi,
\qquad
{\bW_4}=\lim_{\epsilon\rightarrow 0}{\bW_3}|_{\lambda_2=\lambda_1+\epsilon, \
 \widehat{\widetilde{\xi}}\rightarrow {\overline{\xi}}, \
 \widehat{\widetilde{\eta}}\rightarrow {\overline{\eta}}, \
 \widehat{\widetilde{w}}\rightarrow {\overline{w}}
 },
\end{align}
where $\zeta_1$, $ \delta_1^+$, $y_1$ are given by \eqref{transj}, the superscript $'$ denotes the differentiation with respect to ${\lambda_1}$.
Then for any $\bPsi$ satisfying \eqref{skdv_sp}, ${\overline{\bPsi}}$ defined by \eqref {dt4} solves
$\overline{\bPsi}_x=\overline{\bL} \overline{\bPsi}$, $ \overline{\bPsi}_t=\overline{\bT} \overline{\bPsi}$. In particular, the compatibility of \eqref{skdv_sp1}, \eqref {dt4} and
 ${\overline\bPsi}_x={\overline\bL} {\overline\bPsi}$
 leads to
 \begin{subequations}\label{bt4}
\begin{eqnarray}
&&( {\overline{\eta}}-\eta)_x=
	\frac{1}{2}( {\overline{w}}-w)( {\overline{\eta}}+\eta)
	+\frac{( {\overline{w}}_x-w_x)( {\overline{\eta}}-\eta)}{2( {\overline{w}}-w)},\\
	&&( {\overline{\xi}}-\xi)_x=
	\frac{1}{2}( {\overline{w}}-w)( {\overline{\xi}}+\xi)
	+\frac{( {\overline{w}}_x-w_x)( {\overline{\xi}}-\xi)}{2( {\overline{w}}-w)},
	\\
	&&( {\overline{w}}-w)_{xx}=2(\lambda_1+ {\overline{w}}_x+w_x)( {\overline{w}}-w)-\frac{1}{2}( {\overline{w}}-w)^3
+\frac{( {\overline{w}}-w)^2_{x}}{2( {\overline{w}}-w)}
	+\xi {\overline{\eta}}- {\overline{\xi}}\eta.
\end{eqnarray}
\end{subequations}
\end{Proposition}

\begin{Remark}
 In the last two propositions, we may avoid the solutions of the adjoint spectral problem and reformulate the Darboux matrices and transformations for fields in terms of solutions of the spectral problem. To see it, introduce an auxiliary variable \smash{$\delta_j\equiv\frac{\partial^{-1}_x(\xi\psi_j)}{ \psi_j}$} such that~${\delta_j=\delta_j^+}$.
Thus $y_j^+$, $\delta_j^+$ may be replaced by
$
y_j^+=y_j-\delta_j \zeta_j$, $ \delta_j^+=\delta_j$,
and the formulas~\eqref{w2}--\eqref{w4} may be rewritten. Therefore, one may show that \eqref{w1}, \eqref{w2}, \eqref{w4} coincide with the results obtained in \cite{xl}.
\end{Remark}

\begin{Remark}
The BDT given by Proposition~\ref{prop5} survives under the reduction $\xi=\eta$. Indeed, taking $\zeta_1=\delta_1^+$, we have ${\overline{\eta}}= {\overline{\xi}}$.
Therefore from \eqref{w4}--\eqref{bt4}, we directly obtain the BDT for Kupershmidt's sKdV equation \eqref{kuper-skdv} \cite{xl}.
\end{Remark}

As mentioned in last section, system \eqref{gskdv} enjoys another reduction, namely the Geng--Wu's sKdV equation. It is interesting to observe that the BDT presented by Proposition~\ref{prop5} also survives under this reduction. Indeed, we take $\delta_1^+=\eta_x-y_1 \eta+\lambda_1 \zeta_1+\eta\eta_x \zeta_1$. Then, it is easy to check that $\xi=\eta_{xx}-u\eta$ implies ${\overline{\xi}}= {\overline{\eta}}_{xx}- {\overline{u}} {\overline{\eta}}$. In this way, we obtain a BDT for the Geng--Wu's sKdV equation from \eqref{w4}--\eqref{bt4}, and the result is summarized in the following proposition.

\begin{Proposition} \label{prop6}
Let $\xi=\eta_{xx}-u\eta$, and
\begin{align}\label{gw}
& {\overline{w}} =w+H,
\qquad {\overline{\eta}}=\eta+H{\zeta_1},\\
\label{dt5}
&	{\overline{\bPsi}}={\bW_5}\bPsi,
\qquad
{\bW_5}
=
\left(
\begin{matrix}
\lambda-\lambda_1+A_0 H
& H & -H A_1
\\
(\lambda-\lambda_{1})H+
	A_2 A_0 & \lambda-\lambda_{1}+A_2 & -\frac{1}{2}\bigl(H^2+H_x\bigr)A_1
\\
A_0 H{\zeta_1} & H{\zeta_1}
& \lambda- \lambda_1 +A_1H{\zeta_1}
\end{matrix}
\right),
\end{align}
where
\begin{gather*}
 H\equiv - [ y_{1}'+\zeta_{1}' (\eta_x-y_1 \eta+\lambda_1 \zeta_1+\eta\eta_x \zeta_1 )]^{-1},
\\
A_1\equiv\eta_x+\frac{H \eta}{2}
-\frac{H_x\eta}{2H}
+{\lambda_1 {\zeta_1} }
-\frac{\eta_x \eta {\zeta_1}}{2},
\\
 A_0\equiv \frac{1}{2H}\bigl(H^2-H_{x}+H A_1 {\zeta_1} \bigr),
\qquad
 A_2\equiv \frac{1}{2}\bigl(H^2
+{H_{x}}+H A_1 {\zeta_1} \bigr)
,
\end{gather*}
and $\zeta_1$, $y_1$ are given by \eqref{transj}. Then any ${\overline{\bPsi}}$ defined by
\eqref{dt5}, together with \eqref{geng0} and \eqref{skdv_sp},
 solves $\overline{\bPsi}_x=\overline{\bL} \overline{\bPsi}$, $ \overline{\bPsi}_t=\overline{\bT} \overline{\bPsi}$. The corresponding BT, after eliminating $H$ and $\zeta_1$ by \eqref{gw}, becomes
\begin{gather*}
	( {\overline \eta}-\eta)_x =\frac{1}{2}( {\overline w}-w)( {\overline \eta}+\eta)+\frac{( {\overline w}-w)_x}{2( {\overline w}-w)}( {\overline \eta}-\eta),\\
	( {\overline w}-w)_{xx}=2( {\overline w}-w)(\lambda_1+ {\overline w}_x+w_x)
	-\frac{1}{2}( {\overline w}-w)^3
	+\frac{( {\overline w}-w)_x^2}{2( {\overline w}-w)}\\
\phantom{( {\overline w}-w)_{xx}=}{}
	+ \eta_{xx}( {\overline \eta}-\eta)+\eta\eta_x ( {\overline w}-w) +\eta {\overline \eta} (\lambda_1-2w_x).
\end{gather*}
\end{Proposition}

Above BDT is given in terms of variables $w$, $\eta$ and it is possible to go over to variables
$v, \eta$. In fact, by means of $v=2w_x+\eta \eta_x$, ${\overline{v}}=2 {\overline{w}}_x+ {\overline{\eta}} {\overline{\eta}}_x$, from Proposition~\ref{prop6} we obtain
\[
 {\overline{v}}=v+2H_x+H\zeta_1\eta_x
+\frac{1}{2}\zeta_1\eta\bigl(H^2-H_x \bigr),
\]
which should be coupled with the second equation of \eqref{gw}.

\section[B\"acklund and Darboux transformations for gsmKdV equations]{B\"acklund and Darboux transformations\\ for gsmKdV equations}\label{sec5}
In this section, we aim to construct BDTs for gsmKdV equation \eqref{gsmkdv} and consider the relevant reductions.


Let $\bPhi$ be the solution of \eqref{smkdv_spN}. Similarly,
consider a gauge transformation
$
\widetilde{\bPhi}=\bV\bPhi$,
such that
\smash{$
\widetilde{\bPhi}_x=\widetilde{\bU}\widetilde{\bPhi}$}, \smash{$ \widetilde{\bPhi}_{t_N}=\widetilde{\bM}_N\widetilde{\bPhi}$},
where $\widetilde{\bU}$ and $\widetilde{\bM}_N$ are $\bU$ and $\bM_N$ but with $q$, $\zeta$ and $\delta$ replaced by $\widetilde{q}$, $\widetilde{\zeta}$ and $\widetilde{\delta}$, respectively.
The compatibility conditions $\widetilde{\bPsi}_x=\bigl(\widetilde{\bPsi}\bigr)_x$ and $\widetilde{\bPsi}_{t_N}=\bigl(\widetilde{\bPsi}\bigr)_{t_N}$ lead to
\begin{subequations}\label{Vxt}
\begin{align}
\label{Vx}
&\bV_x+\bV \bU- \widetilde{\bU}\bV=0,\\
 &\bV_{t_N}+\bV \bM_N- \widetilde{\bM}_N\bV=0,\label{Vt}
\end{align}
\end{subequations}
respectively.
Therefore, the gauge transformation $\widetilde{\bPhi} = \bV\bPhi$, together with condition \eqref{Vx},
 provides the same Darboux matrix and spatial part of BT for $t_N$-flow of the gsmKdV hierarchy, and from \eqref{Vt} one may obtain the related temporal part of BT. Again in the following we consider $N=3$ case.

We first note that there exists the following gauge transformation for spectral problem \eqref{smkdv_sp1}.
\begin{Proposition} \label{prop7}
 Let
\begin{align}
\widetilde{\bPhi}={\bV}_0\bPhi,
\qquad
{\bV}_0= \left(
\begin{matrix} 0 & 1 & \delta
\\
\lambda & 0 &0
\\
0 &-\widetilde{\zeta} &1+\delta\widetilde{\zeta}
\end{matrix}
\right). \label{dt1_smkdv}
\end{align}
Then \eqref{smkdv_sp1}, \eqref{dt1_smkdv}, and $\widetilde{\bPhi}_{x}=\widetilde{\bU}\widetilde{\bPhi}$ are compatible if and only if
\begin{equation}\label{bt1_smkdv}
\widetilde{y}=-y +\widetilde{\delta} \widetilde{\zeta},\qquad \widetilde{\zeta}_x=-\zeta-\widetilde{y} \widetilde{\zeta},\qquad \widetilde{\delta}=\delta_x+y \delta
\end{equation}
are satisfied.
\end{Proposition}

It is pointed out that there is no free parameter in \eqref{bt1_smkdv} so that the gauge transformation does not qualify a DT. However, it is interesting to note that from \eqref{bt1_smkdv}, the Miura transformations
 \[
 u=2w_x=y_x+y^2 -\xi \zeta, \qquad \eta=\zeta_{x}+y \zeta,\qquad \xi=\delta_{x}+y \delta,
 \]
and
\[
\widetilde{u}=2\widetilde{w}_x=\widetilde{y}_x+\widetilde{y}^2
-\widetilde{\xi }\widetilde{\zeta},
\qquad
\widetilde{\eta}=\widetilde{\zeta}_{x}+\widetilde{y}\widetilde{ \zeta},\qquad
\widetilde{\xi}=\widetilde{\delta}_{x}+\widetilde{y}\widetilde{ \delta},
\]
we have
 \[
\widetilde{w}=w- y,\qquad u=y_x+y^2 +\xi \widetilde{\eta},\qquad \widetilde{\eta}_x=- \eta -y \widetilde{\eta},\qquad
\xi_x= \widetilde{\xi}+y \xi,
 \]
which is exactly \eqref{bt1} if $\lambda_1=0$.

We now present two elementary BDTs for the gsmKdV equation. In the remaining part, for convenience, we introduce $k_j$ by $\lambda_j \equiv k_j^2$ ($j=1, 2$), and let
$q=r_x$, $ q_{[j]}=r_{[j],x}$, $ j=1, 2, 3, 4$.
Our results are summarized as follows.

\begin{Proposition} \label{prop8}
 Let
\begin{align}
{\bPhi}_{[1]}={\bV}_1\bPhi,
\qquad
{\bV}_1= \left(
\begin{matrix}
 \lambda & k_1 p_1+{p_1^2}\Delta_1& 0
\\
 \lambda k_1 p_1^{-1} +\lambda \Delta_1& \lambda &\lambda (\delta-\delta_{[1]})
\\
\lambda (\zeta-\zeta_{[1]}) & 0 &\lambda
\end{matrix}
\right), \label{dt3_smkdv}
\end{align}
where $ p_1={\rm e}^{r_{[1]}+r}$, $ \Delta_1=\frac{1}{2} (\delta_{[1]}-\delta)(\zeta_{[1]}-\zeta)$.
Then \eqref{smkdv_sp1}, \eqref{dt3_smkdv}, ${\bPhi}_{[1],x}={\bU}_{[1]}{\bPhi}_{[1]}$ are compatible if and only if
\begin{subequations}\label{bt3_smkdv}
\begin{align}
(r_{[1]}-r)_x=&2 k_1 \sinh(r_{[1]}+r) +\frac{1}{2} (\delta \zeta_{[1]}-\delta_{[1]}\zeta )+\frac{1}{2}{\rm e}^{2 r_{[1]}+2r} (\delta_{[1]}-\delta )(\zeta_{[1]}-\zeta),\\
(\zeta_{[1]}-\zeta)_x=& k_1 {\rm e}^{-r_{[1]}-r} \zeta _{[1]}-r_{x} (\zeta_{[1]}-\zeta)+\frac{1}{2} \delta_{[1]} \zeta\zeta_{[1]},\\
(\delta_{[1]}-\delta)_x=& k_1 {\rm e}^{-r_{[1]}-r} \delta -r_{[1],x} (\delta_{[1]}-\delta)+\frac{1}{2} \delta \delta_{[1]} \zeta
\end{align}
\end{subequations}
is satisfied. 
\end{Proposition}

\begin{Proposition} \label{prop9}
 Let
\begin{align}
{\bPhi}_{[2]}={\bV}_2\bPhi,
\qquad
{\bV}_2= \left(
\begin{matrix}
 \lambda &- k_2 p_2-{p_2^2}\Delta _2& k_2 p_2 (\delta_{[2]}-\delta)
\\
- \lambda k_2 p_2^{-1}+\lambda \Delta_2& \lambda &\lambda (\delta-\delta_{[2]})
\\
\lambda (\zeta-\zeta_{[2]}) & k_2 p_2 (\zeta_{[2]} -\zeta) &\lambda- k_2^2 +2 k_2 p_2 \Delta_2
\end{matrix}
\right), \label{dt4_smkdv}
\end{align}
where $p_2={\rm e}^{r_{[2]}+r}$, $ \Delta_2=\frac{1}{2} (\delta_{[2]}-\delta)(\zeta_{[2]}-\zeta)$.
Then \eqref{smkdv_sp1}, \eqref{dt4_smkdv}, ${\bPhi}_{[2],x}={\bU}_{[2]}{\bPhi}_{[2]}$ are compatible if and only if
\begin{subequations}\label{bt4_smkdv}
\begin{align}
(r_{[2]}-r)_x=&-2 k_2 \sinh(r_{[2]}+r) +\frac{1}{2} ( \delta \zeta_{[2]}-\delta_{[2]}\zeta)-\frac{1}{2}{\rm e}^{2 r_{[2]}+2r} (\delta_{[2]}-\delta )(\zeta_{[2]}-\zeta),\\
(\zeta_{[2]}-\zeta)_x=&- k_2 {\rm e}^{-r_{[2]}-r} \zeta-r_{[2],x} (\zeta_{[2]}-\zeta)-\frac{1}{2} \delta \zeta_{[2]}\zeta,\\
(\delta_{[2]}-\delta)_x=&- k_2 {\rm e}^{-r_{[2]}-r} \delta_{[2]} -r_{x} (\delta_{[2]}-\delta)-\frac{1}{2} \delta_{[2]}\delta\zeta_{[2]}
\end{align}
\end{subequations}
is satisfied. 
\end{Proposition}

It is noted that the BT \eqref{bt3} may be related to \eqref{bt3_smkdv} (or \eqref{bt4_smkdv}). In fact,
 from \eqref{bt3_smkdv}, the Miura transformation \eqref{miura2} and its deformation
\begin{gather*}
\widehat{\widetilde{u}}=2\widehat{\widetilde{w}}_x=q_{[1],x}+{q}_{[1]}^2
-\frac{1}{2}(\delta_{[1],x}\zeta_{[1]}-\delta _{[1]}\zeta_{[1],x})
,\qquad
\widehat{\widetilde{\eta}}={\zeta}_{[1],x}+{q}_{[1]} { \zeta}_{[1]} ,\\
\widehat{\widetilde{\xi}}={\delta}_{[1],x}+{q}_{[1]} { \delta}_{[1]},
\end{gather*}
 imply{\samepage
 \begin{gather*}
\widehat{\widetilde{w}}=w+ k_1 {\rm e}^{r_{[1]}+r}+\frac{1}{2} {\rm e}^{2 r_{[1]}+2r} (\delta_{[1]}-\delta )(\zeta_{[1]}-\zeta),\qquad
\widehat{\widetilde{\eta}}=\eta+(\widehat{\widetilde{w}}-w) \zeta_{[1]},\\
 \widehat{\widetilde{\xi}}=\xi+(\widehat{\widetilde{w}}-w) \delta,
 \end{gather*}
 then \eqref{bt3} is satisfied when $ \lambda_2=0$.}

To formulate the BDTs given above in terms of the particular solutions of spectral problem~\eqref{smkdv_sp} and adjoint problem
\begin{equation}\label{smkdv_sp3}
\bPhi_{x}^+=-\bPhi^+{\bU},\qquad \bPhi_{t}^+=-\bPhi^+{\bM},
\end{equation}
we
assume that $\bPhi_j=(\phi_{1j},\phi_{2j},\phi_{3j})^{\mathsf{T}}$ is a solution of \eqref{smkdv_sp} at $\lambda=\lambda_j$, $\bPhi_j^+=\bigl(\phi_{1j}^+,\phi_{2j}^+,\phi_{3j}^+\bigr)$ is a solution of \eqref{smkdv_sp3} at $\lambda=\lambda_j$,
where $\bPhi_j\equiv \bPhi_j(x, t; \lambda_j)$, \smash{$\bPhi_j^+\equiv \bPhi_j^+(x, t;\lambda_j)$}; $\phi_{3j}$ and \smash{$\phi_{3j}^+$} are fermionic, and other components are bosonic, for $j=1, 2$. Introduce new variables $b_j$, $\beta_j$, \smash{$b^+_j$}, \smash{$\beta_j^+$} such that
\begin{align}\label{b12}
b_j \equiv -\frac{\phi_{1j}}{\phi_{2j}}, \qquad \beta_j\equiv -\frac{\phi_{3j}}{\phi_{2j}},\qquad
b^+_j\equiv \frac{\phi_{2j}^+}{\phi_{1j}^+}, \qquad \beta^+_j\equiv -\frac{\phi_{3j}^+}{\phi_{1j}^+}.
\end{align}
Those quantities satisfy the following differential equations:
\begin{subequations}\label{b1x}
\begin{align}
& b_{j,x} = \lambda_j b_j^2-1+2 b_j q - \beta_j \delta, \qquad b_{j,x}^+= \lambda_j \bigl(b_j^+\bigr)^2 -1+2 b_j^+ q- \beta_j^+ \zeta,
\\
\label{beta1x}
&\beta_{j,x}=\zeta+\beta_j (q+ \lambda_j b_j)+\frac{1}{2}\beta_j \delta \zeta,\qquad
\beta_{j,x}^+=\delta +\beta_j^+ \bigl(q+ \lambda_j b_j^+\bigr)-\frac{1}{2}\beta_j^+ \delta \zeta,
\end{align}
\end{subequations}
which indicate that we may take $b_j^+ =b_j- \beta_j^+ \beta_j$.

With those preparations, we may rewrite the BDTs given by last propositions and obtain the following result.
\begin{Proposition} \label{prop10}
Let
\begin{align*}
&{\rm e}^{r_{[1]}+r}={ k_1 } \left(b_1-\frac{1}{2} \alpha_1 \beta_1 \right),\qquad
\zeta_{[1]}=\zeta+\frac{ \beta_1}{b_1},\qquad \delta_{[1]}=\delta+\frac{\alpha_1}{b_1};
\\
&{\rm e}^{r_{[2]}+r}={ k_2 } \left(\frac{ 1}{2}\alpha_2^+\beta_2^+-b_2^+\right),\qquad \delta_{[2]}=\delta+\frac{ \beta_2^+}{b_2^+}, \qquad \zeta_{[2]}=\zeta+\frac{ \alpha_2^+}{b_2^+},
\end{align*}
where $b_j$, $\beta_j$, $b_j^+$, $\beta_j^+$ are defined by \eqref{b12}, the auxiliary variables $\alpha_1$, $\alpha_2^+$ satisfy the following first-order differential equations
$
 \alpha_{1,x}=\delta +q \alpha_1-\frac{1}{2}\alpha_1\delta \zeta$, $
 \alpha_{2,x}^+=\zeta +q \alpha_2^++\frac{1}{2}\alpha_2^+ \delta \zeta$,
respectively.
Then \eqref{bt3_smkdv} and \eqref{bt4_smkdv} are satisfied.
\end{Proposition}

Taking the last two BDTs into consideration, we find a compound BDT summarized by the following proposition.

\begin{Proposition} \label{prop11}
 Suppose that $\bPhi$ solves \eqref{smkdv_sp} and denote $k_1 k_2 {\rm e}^{r-r_{[3]}}$ by $C$.
Let	
	\begin{gather}
	\zeta_{[3]}=\zeta+ B\beta_1,\qquad
\delta_{[3]}=\delta+ B\beta_2^{+},\qquad
C=\frac{\lambda_2\lambda_1 B\bigl(b_2^+-b_1+\beta_2^+\beta_1\bigr)}
{ \lambda_2-\lambda_1 },\label{zeta3}\\
{\bPhi}_{[3]}={\bV}_3 \bPhi,\qquad		{\bV}_3 = \left(
	\begin{matrix}
	\lambda- \lambda_2 -\frac{ b_1 (\lambda_2-\lambda_1) }{b_2^+-b_1+\beta_2^+\beta_1}
 &\frac{- b_1b_2^+ (\lambda_2-\lambda_1) }{b_2^+-b_1+\beta_2^+\beta_1}& \frac{ b_1\beta_2^+(\lambda_2-\lambda_1) }{b_2^+-b_1+\beta_2^+\beta_1}
	\\
	\lambda B+\frac{1}{2}\lambda B^2 \beta_2^+\beta_1&\lambda-\frac{ 2 C}{2-B\beta_2^+\beta_1 } &-{\lambda B \beta_2^+}
	\\
	-{\lambda B\beta_1}& -{ \lambda_2 b_2^+ B \beta_1} & \lambda - \lambda_2 \bigl(1+B \beta_2^{+}\beta_1\bigr)
	\end{matrix}
	\right),\nonumber
	\end{gather}
where
\begin{align}\label{B}
 B=\frac{2( \lambda_2 - \lambda_1) }{ 2\bigl( \lambda_2 b_2^+ - \lambda_1 b_1\bigr)+( \lambda_1 + \lambda_2 )\beta_2^{+}\beta_1},
\end{align}
 and $b_1$, $\beta_1$, $b^+_2$, $\beta_2^+$ are given by \eqref{b12}.
Then \smash{${\bPhi}_{[3],x}={\bU}_{[3]}{\bPhi}_{[3]}$}, \smash{${\bPhi}_{[3],t}={\bM}_{[3]}{\bPhi}_{[3]}$} hold.
By eliminating $b_1$, $\beta_1$, $b^+_2$, $\beta_2^+$, the compatibility condition $ \bV_{3,x}+\bV_3\bU- {\bU}_{[3]} \bV_3 =0$ leads to the following~BT:
\begin{subequations}\label{BT3}
\begin{gather}
(r_{[3]}-r)_x=-B-\frac{(C-\lambda_1)(C-\lambda_2)}{BC}-\frac{\delta_{[3]}\zeta-\delta\zeta_{[3]}}{2}\nonumber\\
\phantom{(r_{[3]}-r)_x=}{}-\frac{(\lambda_2-\lambda_1)(\delta_{[3]}-\delta)(\zeta_{[3]}-\zeta)}{2B^2},\label{61a}\\
(\zeta_{[3]}-\zeta)_x=\zeta_{[3]} B-q(\zeta_{[3]}-\zeta)+\frac{(C-\lambda_2)(\zeta_{[3]}-\zeta)}{B} +\frac{\zeta_{[3]} \delta_{[3]}\zeta}{2},\label{61b}\\
(\delta_{[3]}-\delta)_x=\delta_{[3]} B-q(\delta_{[3]}-\delta) +\frac{ (C-\lambda_1)(\delta_{[3]}-\delta)}{B}+\frac{\delta_{[3]}\zeta_{[3]}\delta}{2},\label{61c}
\end{gather}
\end{subequations}
where the formulae of $B$, $B^{-1}$ are given by
\begin{align*}
&B=\frac{m_1}{m_0}+\frac{m_2}{m_1}-\frac{m_0 m_2^2}{m_1^3},\qquad
B^{-1}=\frac{m_0}{m_1}-\frac{m_0^2 m_2}{m_1^3}+\frac{2 m_0^3m_2^2}{m_1^5},
\end{align*}
with
\begin{gather*}
m_0=4(r+r_{[3]})_x-2(\delta_{[3]}\zeta-\delta\zeta_{[3]}),\\
m_1=4C-\frac{2C_{xx}-4r_x C_x+4 \lambda_1\lambda_2}{C}-3r_x (\delta_{[3]}\zeta-\delta\zeta_{[3]})-2\delta_{[3]}\zeta\delta\zeta_{[3]}
\\
\phantom{m_1=}{}+(\delta_{[3]}-\delta) \zeta_x-\delta_x (\zeta_{[3]}-\zeta),
\\
m_2=3C(\delta_{[3]}\zeta-\delta\zeta_{[3]})-3\lambda_1(\delta_{[3]}-\delta)\zeta +3 \lambda_2\delta(\zeta_{[3]}-\zeta) +(\lambda_2-\lambda_1)(\delta_{[3]}-\delta)(\zeta_{[3]}-\zeta).
\end{gather*}
\end{Proposition}
\begin{proof}
 It is tedious but straightforward to check that
 ${\bPhi}_{[3],x}={\bU}_{[3]}{\bPhi}_{[3]}$, ${\bPhi}_{[3],t}={\bM}_{[3]}{\bPhi}_{[3]}$ hold, and the compatibility condition $ \bV_{3,x}+\bV_3\bU- {\bU}_{[3]} \bV_3 =0$ together with \eqref{zeta3} leads to~\eqref{b1x}.
 Now we show that the BT \eqref{BT3} may be obtained from \eqref{zeta3} and \eqref{b1x}.
Firstly, \eqref{zeta3} and~\eqref{B} lead to
 \begin{gather*}
 \beta_1=\frac{\zeta_{[3]}-\zeta}{B},
\qquad
\beta_2^+=\frac{ \delta_{[3]}-\delta }{B},\\ b_1=\frac{1}{B}-\frac{C}{\lambda_1 B}+
\frac{\beta_2^+ \beta_1}{2 },
\qquad
 b_2^+
=\frac{1}{B}-\frac{C}{\lambda_2 B}-\frac{\beta_2^+ \beta_1}{2 }. \label{b1}
\end{gather*}
Plugging above into the first-order differential equations \eqref{b1x}, one may obtain (\ref{BT3}) and the following equation:
\begin{gather}
B_x=2C-\lambda_1-\lambda_2-2r_x B
+B^2+\frac{1}{2}B(\delta_{[3]}\zeta-\delta\zeta_{[3]})\nonumber\\
\phantom{B_x=}{}
+\frac{1}{2B}(\lambda_2-\lambda_1)(\delta_{[3]}-\delta)(\zeta_{[3]}-\zeta).\label{Bx}
\end{gather}
The auxiliary variable $B$ may be eliminated as follows.
We first differentiate \eqref{61a} with respect to $x$, and meanwhile use \eqref{61b}, \eqref{61c} and \eqref{Bx} to replace $\zeta_{[3],x}$, $\delta_{[3],x}$, $B_x$, thus get another algebraic equation for $B$. We then use \eqref{61a} to simplify this equation and find
$
m_0B=m_1+\frac{m_2}{B}$.
As $m_2^2= 18 (C-\lambda_1)(C-\lambda_2)\zeta_{[3]}\delta_{[3]}\zeta\delta$, $m_2^3=0,$ we may solve the last equation and obtain $B$.
\end{proof}

\begin{Remark}Setting the fermionic variables to zero, we recover the generalized BT for the mKdV equation
\begin{gather*}
[ (r_{[3]}-r)_{xx}-4 k_1 k_2 \sinh{(r_{[3]}-r)}]^2\\
\qquad=
 (r_{[3]}+r)_x^2\big[(r_{[3]}-r)_x^2+4\bigl( k_1^2+ k_2^2\bigr)-8 k_1 k_2 \cosh{(r_{[3]}-r)} \big],
\end{gather*}
which was referred as the type-II BT for the mKdV hierarchy in \cite{grz} (see also \cite{ml}).
\end{Remark}

Now we consider the limit of the compound BDT and construct one more BDT. To this end,
let $\lambda_2=\lambda_1+ \epsilon$ and
$\bPhi_2^+= \bPhi_1^+(x, t;\lambda_1+ \epsilon)$, then
 $b_2^+=b_1^+ (x, t; \lambda_1 +\epsilon)$, $ \beta_2^+=\beta_1^+( x, t;\lambda_1 +\epsilon).$
Choosing
\smash{$
b_1^+ =b_1- \beta_1^+ \beta_1$},
 the limit of $B$ is found to be
\begin{equation}\label{BK}
\lim_{\epsilon \rightarrow 0} B=K\equiv
 \frac{2}{ 2b_1-\beta_1^{+}\beta_1+2\lambda_1 \bigl( b_1'-\beta_1^{+}\beta_1'\bigr) },
\end{equation}
where $b_1'=\frac{\partial b_1}{\partial \lambda_1}$, $\beta_1'=\frac{\partial \beta_1}{\partial \lambda_1}$.
Thus after taking the limit of the last property as $\epsilon \rightarrow 0$, from Proposition~\ref{prop11} we have the following proposition.

\begin{Proposition} \label{prop12}
Suppose that $\bPhi$ satisfies \eqref{smkdv_sp}, and let
\begin{align*}
&\zeta_{[4]}= \zeta+ K \beta_1,\qquad
\delta_{[4]}=\delta+ K \beta_1^{+},\qquad
 {\rm e}^{r-r_{[4]}}={\lambda_1K \bigl(b_1'-\beta_1^+\beta_1'\bigr) },\\
&{\bPhi}_{[4]}={\bV}_4\bPhi,\qquad
		{\bV}_4= \left(
	\begin{matrix}
	\lambda- \lambda_1 -\frac{ b_1 }{b_1'-\beta_1^+\beta_1' } &\frac{- b_1(b_1- \beta_1^+ \beta_1)}{b_1'-\beta_1^+\beta_1'}& \frac{ b_1\beta_1^+}{b_1'-\beta_1^+\beta_1'}
	\\
	\lambda K +\frac{\lambda }{2} K ^2 \beta_1^+\beta_1&\lambda-\frac{2 \lambda_1 {\rm e}^{r-r_{[4]}}}{2-{K} \beta_1^+\beta_1} &-{\lambda K \beta_1^+}
	\\
	-{\lambda K \beta_1}& -{ \lambda_1 b_1 K \beta_1} & \lambda- \lambda_1\bigl(1+ K \beta_1^{+}\beta_1\bigr)
	\end{matrix}
	\right),
	\end{align*}
where $b_1$, $\beta_1$, $\beta_1^+$ are defined by \eqref{b12}, $K$ is given in \eqref{BK}. Then
${\bPhi}_{[4],x}={\bU}_{[4]}{\bPhi}_{[4]}$, ${\bPhi}_{[4],t}={\bM}_{[4]}{\bPhi}_{[4]}$ hold.
The compatibility condition $ \bV_{4,x}+\bV_4\bU- {\bU}_{[4]} \bV_4 =0$, after eliminating $b_1$, $\beta_1$, \smash{$\beta_1^+$}, leads to \eqref{BT3}
 with the replacement $(\zeta_{[3]}\rightarrow \zeta_{[4]}, \delta_{[3]}\rightarrow \delta_{[4]}, r_{[3]}\rightarrow r_{[4]}, k_2\rightarrow k_1)$.
\end{Proposition}

As Hu's smKdV equation is related to the gsmKdV equation, it is possible to obtain its BDTs from the above results. Indeed, taking account of the transformations
$ y=q+\frac{1}{2}\delta\zeta$, $\xi=\delta_{x}+q \delta$, $y_{[j]}=q_{[j]}+\frac{1}{2}\delta_{[j]}\zeta_{[j]}$, $
 \xi_{[j]}=\delta_{[j],x}+q_{[j]} \delta_{[j]}$,
and Propositions \ref{prop8}--\ref{prop12},
we have
\begin{gather*}
\xi_{[1]}=\xi+{\lambda_1 b_1\delta} ,\qquad y_{[1]}=y+\lambda_1 b_1+\frac{\delta\beta_1-1}{b_1};\\
\xi_{[2]}=\xi+{\lambda_2\bigl( {b_2^+} \delta+\beta_2^+\bigr)},\qquad y_{[2]}=y+\lambda_2 b_2^++\frac{\delta \alpha_2^+-1}{b_2^+}
+\frac{\beta_2^+}{\bigl(b_2^+\bigr)^2}; \\
\xi_{[3]}=\xi
+\frac{(\lambda_1- \lambda_2) \bigl(b_2^+\delta+\beta_2^+\bigr)b_1}{ b_2^+- b_1 +\beta_2^+\beta_1},\\
 y_{[3]}=y
+\frac{(\lambda_1- \lambda_2) b_1 b_2^+}{ b_2^+- b_1 +\beta_2^+\beta_1}
+\frac{(\lambda_1- \lambda_2) (1-\delta \beta_1)}{\lambda_2 b_2^+- \lambda_1 b_1 +\lambda_2\beta_2^+\beta_1};\\
\xi_{[4]}=\xi-\frac{b_1\bigl(b_1\delta-\beta_1^+\beta_1\delta+\beta_1^+\bigr)}{b_1'-\beta_1^+\beta_1'},\\
y_{[4]}=y
-\frac{ b_1 \bigl(b_1-\beta_1^+\beta_1\bigr)}{b_1'-\beta_1^+\beta_1'}
-\frac{(1-\delta \beta_1)}{b_1+\lambda_1 \bigl(b_1'-\beta_1^+\beta_1'\bigr)}.
\end{gather*}
Hence, these formulae, together with the transformations of $\zeta_{[j]}$ and Propositions \ref{prop8}--\ref{prop12}, also yield the BDTs for Hu's smKdV equation \eqref{smkdv1}.
Since its DTs have already been considered in~\cite{zty} under a different spectral problem, how to relate their results to ours is a problem deserving further study.

Finally, we can show that the BDT presented in last proposition survives under the two reductions mentioned in the introduction.
For the first reduction $\zeta=\delta$, its BDT is described by the following proposition.
\begin{Proposition} \label{prop13}
Suppose that $\bPhi$ satisfies \eqref{smkdv_sp} with $\zeta=\delta$. Let
 \begin{gather*}
{\bPhi}_{[4]}={\bV}_5\bPhi,\qquad {\bV}_5= {\bV}_4|_{\beta_{1}^+=\beta_1},\nonumber \\
\zeta_{[4]}= \zeta+ \frac{\beta_1}{ b_1+\lambda_1 ( b_1'-\beta_1 \beta_1') } ,\qquad
 {\rm e}^{r-r_{[4]}}=1- \frac{ b_1}{ b_1+\lambda_1 ( b_1'-\beta_1 \beta_1') },
\end{gather*}
where
$b_1$, $\beta_1$ are defined in \eqref{b12}.
Then
${\bPhi}_{[4],x}={\bU}_{[4]}{\bPhi}_{[4]}$, ${\bPhi}_{[4],t}={\bM}_{[4]}{\bPhi}_{[4]}$ hold.
The corresponding
BT is
\begin{align*}
&(r_{[4]}-r)_x=-{R}-\frac{\lambda_1 ({\rm e}^{r-r_{[4]}} -1)^2}{{R} {\rm e}^{r-r_{[4]}}}+\zeta\zeta_{[4]},\\
&(\zeta_{[4]}-\zeta)_x=\zeta_{[4]} {R}-r_x (\zeta_{[4]}-\zeta)+\frac{\lambda_1 ( {\rm e}^{r-r_{[4]}}-1 )(\zeta_{[4]}-\zeta)}{{R}},
\end{align*}
where
\begin{align*}
&{R} =\frac{m_3}{2(r+ r_{[4]})_x+2\zeta\zeta_{[4]} }+\frac{3\lambda_1 ( 1- {\rm e}^{r-r_{[4]}})\zeta\zeta_{[4]}}{m_3},\\
&m_3=( r_{[4]}-r )_{xx}+r_x^2-( r_{[4],x})^2
+2\lambda_1( {\rm e}^{r-r_{[4]}}- {\rm e}^{r_{[4]}-r } )
+3 r_x \zeta\zeta_{[4]}
+(\zeta_{[4]}-\zeta) \zeta_x.
\end{align*}
\end{Proposition}

In particular, we have a BDT for the Kupershmidt's smKdV equation \eqref{kuper-smkdv}. Indeed, let~${R_0\equiv \phi_{21}\phi_{11}'-\phi_{11}\phi_{21}'+\phi_{31}\phi_{31}'}$ and notice that $R_{0,x}=-\phi_{11}^2$ and $ b_1'-\beta_1 \beta_1'=-R_0/\phi_{21}^2$ hold. Then the DT presented in last proposition, after being rewritten,
coincides with the DT in \cite{ztl,zty}.

For the second reduction, our result is summarized as follows.
\begin{Proposition} \label{prop14}
Suppose that $\bPhi$ satisfies \eqref{smkdv_sp} with $\delta= \zeta_{xx}+\bigl(q_x-q^2\bigr) \zeta-\frac{1}{2}\zeta_{xx} \zeta_x\zeta$. Let
 \begin{gather*}
{\bPhi}_{[4]}={\bV}_6\bPhi,\qquad
{\bV}_6={\bV}_4|_{\beta_{1}^+=\zeta_x+\lambda_1\beta_1+\zeta(q+\lambda_1 b_1)},
 \\
 \zeta_{[4]} = \zeta+ {S} \beta_1,\qquad
 {\rm e}^{r-r_{[4]}}={\lambda_1{S} [b_1'+\beta_1'(\zeta_x+\lambda_1\beta_1+\zeta(q+\lambda_1 b_1))] },
\end{gather*}
where
\begin{equation*}
{S}=
 \frac{2}{ 2b_1+2\lambda_1 b_1'+(\beta_1+2\lambda_1\beta_1')(\zeta_x+\lambda_1\beta_1+\zeta(q+\lambda_1 b_1))},
\end{equation*}
$b_1$ and $\beta_1$ are defined by \eqref{b12}.
Then
${\bPhi}_{[4],x}={\bU}_{[4]}{\bPhi}_{[4]}$, ${\bPhi}_{[4],t}={\bM}_{[4]}{\bPhi}_{[4]}$ hold.
The related
BT is given by
\begin{gather*}
q_{[4]}=r_{[4],x}=q-{S}-\frac{\lambda_1 ({\rm e}^{r-r_{[4]}} -1)^2}{{S} {\rm e}^{r-r_{[4]}}}
-\frac{1}{2}{S}\zeta_x \zeta\\
\phantom{q_{[4]}=r_{[4],x}=}{}
-\frac{1}{2} (\zeta_{[4]}- \zeta) \big[\zeta_{xx}+\zeta\bigl(\lambda_1+q_x-q^2\bigr)\big]
+\frac{1}{4} \zeta_{[4]}\zeta_{xx}\zeta_x \zeta,
\\
\zeta_{[4],x}=\zeta_x+\zeta_{[4]} {S}-q(\zeta_{[4]}-\zeta)
+\frac{\lambda_1}{{S}} ({\rm e}^{r-r_{[4]}}-1)(\zeta_{[4]}-\zeta)
+\frac{1}{2}\zeta \zeta_{[4]} (\zeta_{xx}+{S}\zeta_x),\\
{S}_x=2\lambda_1( {\rm e}^{r-r_{[4]}}-1)
-{S} q-{S} r_{[4],x} -\frac{\lambda_1 ({\rm e}^{r-r_{[4]}} -1)^2}{ {\rm e}^{r-r_{[4]}}}.
\end{gather*}
\end{Proposition}
In this way, we succeed in obtaining a BDT for equation \eqref{m-geng}.

\section{Exact solutions}\label{sec6}
As we know that Darboux transformations and B\"acklund transformations can be employed to generate solutions for the associated nonlinear equations. In this section, we apply the results obtained in last two sections to construct exact solutions of the sKdV and the smKdV systems. To do it, we initiate from a trivial solution of an evolution equation along with the corresponding spectral problem or Riccati type equation, then employ algebraic and differential operations to construct new solutions.

\subsection{Solutions of gsmKdV equation and its reductions}
We start from the zero solution of gsmKdV equation: $q=0$, $\zeta=0$, $\delta=0$. Then by solving~\eqref{b1x} and the differential equations in Propositions \ref{prop10}, we obtain
\begin{gather*}
b_1 =-\frac{1}{ k_1}\coth P_1 ,\qquad \alpha_1 =\gamma_1, \qquad \beta_1 =\theta_1{\rm csch } P_1,
\qquad \beta_1^+ =\mu_1{\rm csch } P_1,
\\
b_2^+=-\frac{1}{ k_2}\tanh P_2 ,\qquad \alpha_2^+=\gamma_2, \qquad \beta_2^+=\mu_2{\rm sech } P_2,
\end{gather*}
where $P_j= k_jx+4 k_j^3t+c_j$, $\lambda_j=k_j^2$, and $c_j$ are arbitrary bosonic constants while $ \gamma_j$, $ \theta_j$, $\mu_j$ are arbitrary fermionic constants, for $ j=1, 2$.

Therefore, Propositions \ref{prop10}--\ref{prop12} give the following three exact solutions:
 \begin{align*}
&\zeta_{[1]}=- k_1\theta_1{\rm sech} P_1,\qquad \delta_{[1]}=- k_1\gamma_1\tanh P_1,\qquad
q_{[1]}=\frac{-2k_1+k_1^2 \theta_1\gamma_1\cosh {P_1}}{\sinh{2P_1}-k_1 \theta_1\gamma_1\sinh {P_1}};
\\
	&\zeta_{[3]}=\frac{2\bigl( k_2^2- k_1^2\bigr)\theta_1 \cosh P_2}{( k_2+ k_1)\cosh (P_2-P_1)-( k_2- k_1)\cosh (P_2+P_1)},\\
	&\delta_{[3]}=\frac{2\bigl( k_2^2- k_1^2\bigr)\mu_2 \sinh P_1}{( k_2+ k_1)\cosh (P_2-P_1)-( k_2- k_1)\cosh (P_2+P_1)},\\
&q_{[3]}=\partial_x\ln
\left [	 \frac{( k_2+ k_1)\cosh (P_2-P_1)-( k_2- k_1)\cosh (P_2+P_1)+\bigl( k_1^2+ k_2^2\bigr) \mu_2\theta_1}
	 {( k_2+ k_1)\cosh (P_2-P_1)+( k_2- k_1)\cosh (P_2+P_1)+2 k_1 k_2\mu_2\theta_1}
\right ];
\\
&\zeta_{[4]}=\frac{4k_1\theta_1\sinh P_1}{
2k_1 \bigl(x+12 k_1^2 t\bigr)-\sinh{2 P_1}
	},
\qquad
\delta_{[4]}=\frac{4k_1\mu_1\sinh P_1}{
2k_1 \bigl(x+12 k_1^2 t\bigr)-\sinh{2 P_1}},\\
&q_{[4]}=\partial_x\ln
\left[
	 \frac{
	 2k_1 \bigl(x+12 k_1^2 t\bigr)- \sinh{2 P_1}
	 +2k_1 \mu_1\theta_1 \sinh^2 {P_1}
	 }{
	2k_1 \bigl(x+12 k_1^2 t\bigr)+ \sinh{2 P_1}
	 -2k_1 \mu_1\theta_1 \cosh^2 {P_1}}
\right].
 \end{align*}
 One may easily get the solutions of Hu's smKdV equation by
$ y_{[j]}=q_{[j]}+\frac{1}{2}\delta_{[j]}\zeta_{[j]}$, $
 \xi_{[j]}=\delta_{[j],x}+q_{[j]} \delta_{[j]}
 $
 together with $\zeta_{[j]}$.

From Propositions \ref{prop13} and \ref{prop14} or from the reductions of $(\zeta_{[4]}, \delta_{[4]}, q_{[4]})$, we also obtain a~common solution for Kupershmidt's smKdV and Geng--Wu's smKdV equations, that is
 \begin{align*}
&\zeta_{[4]}=\frac{4k_1\theta_1\sinh P_1}{
2k_1 \bigl(x+12 k_1^2 t\bigr)-\sinh{2 P_1}
	},\qquad
	q_{[4]}=\partial_x\ln
\left [
	 \frac{2k_1 \bigl(x+12 k_1^2 t\bigr)- \sinh{2 P_1}}{
	2k_1 \bigl(x+12 k_1^2 t\bigr)+ \sinh{2 P_1}
	}
\right ].
 \end{align*}

Let fermionic variables disappear and $Q_1=P_1-\frac{\text{i}\pi}{4}$. We have
\begin{align*}
q_{[1]}=\frac{2 \text{i}k_1 }{\cosh 2Q_1},\qquad
q_{[4]}
=\frac{4\text{i} k_1 \big[\cosh 2Q_1-2 k_1\bigl(x+12 k_1^2 t\bigr)\sinh2 Q_1\big]}{4 k_1^2\bigl(x+12 k_1^2 t\bigr)^2+\cosh^2 2Q_1},
\end{align*}
which are two pure imaginary
 solutions of mKdV equation.

\subsection{Solutions of gsKdV equation and its reductions}
First choose the initial solution $w=0$, $\xi=0$, $\eta=0$. Then solving the differential equations~\eqref{y1x} yields
$
y_1 =k_1\tanh P_1$, $
\zeta_1 =\theta_1 {\rm sech} P_1$, $ \delta_1^+=\mu_1 {\rm sech} P_1$;
$
y_2^+=k_2\coth P_2$, $
\delta_2^+=\mu_2 {\rm csch} P_2$,
where $ P_j=k_j\bigl(x+4k_j^2 t+c_j\bigr)$, $\lambda_j=k_j^2$, and $c_j$ are arbitrary fermionic constants while~$\mu_j$ and $ \theta_j$ are arbitrary fermionic constants, for $ j=1, 2$.

Therefore, Proposition~\ref{prop4} leads to a one-soliton and two-soliton solutions of gsKdV equation. They are given by
\begin{gather*}
\widetilde{u}=-2k_1^2 {\rm sech^2} P_1,\qquad \widetilde{\eta}=-\theta_1 {\rm sech} P_1,\qquad \widetilde{\xi}=0;
\\
\widehat{\widetilde{u}}
=-2 \partial_x^2 \ln Q,\qquad
\widehat{\widetilde{\eta}}
=\frac{2\bigl(k_1^2-k_2^2\bigr)\theta_1\sinh P_2}{Q},\qquad
\widehat{\widetilde{\xi}}
=\frac{2\bigl(k_1^2-k_2^2\bigr)\mu_2\cosh P_1}{Q},
 \end{gather*}
where
$
 Q={(k_2+k_1)\cosh(P_2-P_1)+(k_2-k_1)\cosh(P_2+P_1)+2\mu_2\theta_1}$.
 From Proposition~\ref{prop5}, one has the solution of gsKdV equation
 \begin{align*}
 &{\overline{u}}=-2 \partial_x^2 \ln \big [2k_1^2\bigl(x+12 k_1^2 t\bigr)+ k_1 \sinh 2P_1 -2 \mu_1 \theta_1\sinh^2 P_1\big ],\\
& {\overline{\eta}}=-\frac{4k_1\theta_1\cosh P_1}{2k_1 \bigl(x+12 k_1^2 t\bigr)+\sinh 2P_1 },\qquad
 {\overline{\xi}}=-\frac{4k_1\mu_1\cosh P_1}{2k_1 \bigl(x+12 k_1^2 t\bigr)+\sinh 2P_1 }.
 \end{align*}

Similarly, by the reductions of $({\overline{\eta}}, {\overline{\xi}}, {\overline{u}})$, we also obtain the common solution for Kupershmidt's sKdV and Geng--Wu's sKdV equations, that is
 \begin{align*}
{\overline{\eta}}=-\frac{4k_1\theta_1\cosh P_1}{2k_1 \bigl(x+12 k_1^2 t\bigr)+\sinh 2P_1 },\qquad
{\overline{v}}=-2 \partial_x^2 \ln \big [ 2k_1 \bigl(x+12 k_1^2 t\bigr) +\sinh 2P_1\big ].
 \end{align*}

We also note that by combining the corresponding Miura transformation and Galilean transformation one may generate other solutions for gsKdV, Kupershmidt's and Geng--Wu's sKdV equations from the solutions of smKdV equation $(\zeta_{[j]}, \delta_{[j]}, q_{[j] }, j \ge 3)$.

\section{Conclusions and discussions}\label{sec7}

In this paper, first by constructing a Miura transformation, we obtained a gsmKdV equation from the gsKdV equation, and we showed that
not only the Kupershmidt's, but also Geng--Wu's sKdV together with the corresponding smKdV equations are reductions of them.
We also provide the first negative flows of the gsKdV and gsmKdV hierarchies.

From the binary BDT of the gsKdV equation, we obtained a BDT of the Geng--Wu's sKdV equation which is given in Proposition~\ref{prop6}.
We also constructed the DBTs for gsmKdV equation, and five transformations (one non-parameter transformations, two elementary BDTs, one compound BDT and one generalized BDT) were obtained which are presented in Proposi\-tions~\mbox{\ref{prop7}--\ref{prop14}}. We showed that the non-parameter transformation is related to the elementary DBT of gsKdV equation, while the elementary BTs are related to the generalized binary BT of gsKdV equation.
 The generalized DBT was used to obtain the DBTs for Geng--Wu's smKdV and Kupershmidt's smKdV equations, respectively, see Propositions~\ref{prop13} and~\ref{prop14}.
It was observed that the DBTs for Hu's smKdV equation can be obtained from Propositions~\ref{prop8}--\ref{prop12} directly.

Since all flows of the sKdV or smKdV hierarchy share the same spatial parts of spectral problem and Miura transformation, thus Darboux matrices and spatial parts of BTs can be used for any flow of the hierarchy.

There are some interesting problems to be explored. For example, it is interesting to construct the $N$-fold Darboux transformations and multi-soliton solutions for gsKdV and gsmKdV hierarchies,
derive the
temporal parts of BTs and solutions for the super sinh-Gordon equation~\eqref{sinh},
 and explore the dispersionless limits of gsKdV and gsmKdV equations.
Besides, the Hamiltonian structures and integrable discretizations for the gsmKdV equation are also worthy of study.
Progress in these directions may be published elsewhere.

\subsection*{Acknowledgements}

We thank the anonymous referees for their useful comments.
This work is supported by the National Natural Science Foundation of China (Grant Nos. 12175111, 11931107 and 12171474),
and NSFC-RFBR (Grant No. 12111530003) and the K.C.~Wong Magna Fund in Ningbo University.

\pdfbookmark[1]{References}{ref}
\LastPageEnding


\begin{thebibliography}{99}
\footnotesize\itemsep=0pt

\bibitem{pa}
Adamopoulou P., Papamikos G., Entwining {Y}ang--{B}axter maps over {G}rassmann
 algebras,
 \href{https://doi.org/10.1016/j.physd.2024.134469}{\textit{Phys.~D}}
 \textbf{472} (2025), 134469, 9~pages,
 \href{http://arxiv.org/abs/2311.18673}{arXiv:2311.18673}.

\bibitem{aaetal}
Adans Y.F., Aguirre A.R., Gomes J.F., Lobo G.V., Zimerman A.H., S{K}d{V},
 {S}m{K}d{V} flows and their supersymmetric gauge-{M}iura transformations,
 \href{https://doi.org/10.46298/ocnmp.13294}{\textit{Open Commun. Nonlinear
 Math. Phys.}} (2024), 65--86,
 \href{http://arxiv.org/abs/2403.16285}{arXiv:2403.16285}.

\bibitem{yfa}
Adans Y.F., Fran\c{c}a G., Gomes J.F., Lobo G.V., Zimerman A.H., Negative flows
 of generalized {K}d{V} and m{K}d{V} hierarchies and their gauge-{M}iura
 transformations,
 \href{https://doi.org/10.1007/jhep08(2023)160}{\textit{J.~High Energy Phys.}}
 \textbf{2023} (2023), no.~8, 160, 40~pages,
 \href{http://arxiv.org/abs/2304.01749}{arXiv:2304.01749}.

\bibitem{ara}
Aguirre A.R., Retore A.L., Gomes J.F., Spano N.I., Zimerman A.H., Defects in
 the supersymmetric m{K}d{V} hierarchy via {B}\"acklund transformations,
 \href{https://doi.org/10.1007/jhep01(2018)018}{\textit{J.~High Energy Phys.}}
 \textbf{2018} (2018), no.~1, 018, 38~pages,
 \href{http://arxiv.org/abs/1709.05568}{arXiv:1709.05568}.

\bibitem{bc}
Babalic C.N., Carstea A.S., Bilinear approach to {K}uperschmidt super-{K}d{V}
 type equations,
 \href{https://doi.org/10.1088/1751-8121/aabda5}{\textit{J.~Phys.~A}}
 \textbf{51} (2018), 225204, 9~pages,
 \href{http://arxiv.org/abs/1712.06854}{arXiv:1712.06854}.

\bibitem{ber}
Berezin F.A., Introduction to superanalysis, \textit{Math. Phys. Appl. Math.},
 Vol.~9, \href{https://doi.org/10.1007/978-94-017-1963-6}{D.~Reidel Publishing
 Co.}, Dordrecht, 1987.

\bibitem{CK}
Chaichian M., Kulish P.P., On the method of inverse scattering problem and
 B\"acklund transformations for supersymmetric equations,
 \href{https://doi.org/10.1016/0370-2693(78)90473-2}{\textit{Phys. Lett.~B}}
 \textbf{78} (1978), 413--416.

\bibitem{bfg}
Gao B., Tian K., Liu Q.P., Some super systems with local bi-{H}amiltonian
 operators,
 \href{https://doi.org/10.1016/j.physleta.2018.11.011}{\textit{Phys. Lett.~A}}
 \textbf{383} (2019), 400--405.

\bibitem{bfgk}
Gao B., Tian K., Liu Q.P., A super {D}egasperis--{P}rocesi equation and related
 integrable systems,
 \href{https://doi.org/10.1098/rspa.2020.0780}{\textit{Proc.~R. Soc.~A}}
 \textbf{477} (2021), 20200780, 17~pages.

\bibitem{yyg}
Ge Y., Zuo D., A new class of {E}uler equation on the dual of the~{$N=1$}
 extended {N}eveu--{S}chwarz algebra,
 \href{https://doi.org/10.1063/1.5051755}{\textit{J.~Math. Phys.}} \textbf{59}
 (2018), 113505, 8~pages.

\bibitem{gw}
Geng X., Wu L., A new super-extension of the {K}d{V} hierarchy,
 \href{https://doi.org/10.1016/j.aml.2010.02.014}{\textit{Appl. Math. Lett.}}
 \textbf{23} (2010), 716--721.

\bibitem{xbl}
Geng X., Xue B., Wu L., A super {C}amassa--{H}olm equation with {$N$}-peakon
 solutions,
 \href{https://doi.org/10.1111/j.1467-9590.2012.00555.x}{\textit{Stud. Appl.
 Math.}} \textbf{130} (2013), 1--16,
 \href{http://arxiv.org/abs/9590.2012}{arXiv:9590.2012}.

\bibitem{GS}
Girardello L., Sciuto S., Inverse scattering-like problem for supersymmetric
 models, \href{https://doi.org/10.1016/0370-2693(78)90703-7}{\textit{Phys.
 Lett.~B}} \textbf{77} (1978), 267--269.

\bibitem{grz}
Gomes J.F., Retore A.L., Zimerman A.H., Construction of type-{II} {B}\"acklund
 transformation for the m{K}d{V} hierarchy,
 \href{https://doi.org/10.1088/1751-8113/48/40/405203}{\textit{J.~Phys.~A}}
 \textbf{48} (2015), 405203, 19~pages,
 \href{http://arxiv.org/abs/1505.01024}{arXiv:1505.01024}.

\bibitem{ggg}
Grahovski G.G., Konstantinou-Rizos S., Mikhailov A.V., Grassmann extensions of
 {Y}ang--{B}axter maps,
 \href{https://doi.org/10.1088/1751-8113/49/14/145202}{\textit{J.~Phys.~A}}
 \textbf{49} (2016), 145202, 17~pages,
 \href{http://arxiv.org/abs/1510.06913}{arXiv:1510.06913}.

\bibitem{gg}
Grahovski G.G., Mikhailov A.V., Integrable discretisations for a class of
 nonlinear {S}chr\"odinger equations on {G}rassmann algebras,
 \href{https://doi.org/10.1016/j.physleta.2013.10.018}{\textit{Phys. Lett.~A}}
 \textbf{377} (2013), 3254--3259,
 \href{http://arxiv.org/abs/1303.1853}{arXiv:1303.1853}.

\bibitem{mg1}
G\"urses M., O\v{g}uz \"O., A super {AKNS} scheme,
 \href{https://doi.org/10.1016/0375-9601(85)90033-7}{\textit{Phys. Lett.~A}}
 \textbf{108} (1985), 437--440.

\bibitem{mg2}
G\"urses M., O\v{g}uz \"O., A super soliton connection,
 \href{https://doi.org/10.1007/BF00400221}{\textit{Lett. Math. Phys.}}
 \textbf{11} (1986), 235--246.

\bibitem{hp}
Holod P.I., Pakuliak S.Z., On the superextension of the
 {K}adomtsev--{P}etviashvili equation and finite-gap solutions of
 {K}orteweg--de {V}ries superequations, in Problems of {M}odern {Q}uantum
 {F}ield {T}heory, \textit{Res. Rep. Phys.},
 \href{https://doi.org/10.1007/978-3-642-84000-5_8}{Springer}, Berlin, 1989,
 107--116.

\bibitem{Hr}
Hrub\'y J., On the supersymmetric sine-{G}ordon model and a two-dimensional
 bag, \href{https://doi.org/10.1016/0550-3213(77)90373-X}{\textit{Nuclear
 Phys.~B}} \textbf{131} (1977), 275--284.

\bibitem{hu}
Hu X.-B., An approach to generate superextensions of integrable systems,
 \href{https://doi.org/10.1088/0305-4470/30/2/023}{\textit{J.~Phys.~A}}
 \textbf{30} (1997), 619--632.

\bibitem{IK}
Inami T., Kanno H., Lie superalgebraic approach to super {T}oda lattice and
 generalized super {K}d{V} equations,
 \href{https://doi.org/10.1007/BF02099072}{\textit{Comm. Math. Phys.}}
 \textbf{136} (1991), 519--542.

\bibitem{ker}
Kersten P.H.M., Symmetries for the super modified {K}d{V} equation,
 \href{https://doi.org/10.1063/1.528193}{\textit{J.~Math. Phys.}} \textbf{29}
 (1988), 2187--2189.

\bibitem{kg}
Kersten P.H.M., Gragert P.K.H., Symmetries for the super-{K}d{V} equation,
 \href{https://doi.org/10.1088/0305-4470/21/11/002}{\textit{J.~Phys.~A}}
 \textbf{21} (1988), L579--L584.

\bibitem{sk}
Konstantinou-Rizos S., On the~{$3D$} consistency of a {G}rassmann extended
 lattice {B}oussinesq system,
 \href{https://doi.org/10.1016/j.nuclphysb.2019.114878}{\textit{Nuclear
 Phys.~B}} \textbf{951} (2020), 114878, 24~pages,
 \href{http://arxiv.org/abs/1908.00565}{arXiv:1908.00565}.

\bibitem{skr}
Konstantinou-Rizos S., Kouloukas T.E., A noncommutative discrete potential
 {K}d{V} lift, \href{https://doi.org/10.1063/1.5041947}{\textit{J.~Math.
 Phys.}} \textbf{59} (2018), 063506, 13~pages,
 \href{http://arxiv.org/abs/1611.08923}{arXiv:1611.08923}.

\bibitem{ppk}
Kulish P.P., Quantum {${\rm osp}$}-invariant nonlinear {S}chr\"odinger
 equation, \href{https://doi.org/10.1007/BF00704591}{\textit{Lett. Math.
 Phys.}} \textbf{10} (1985), 87--93.

\bibitem{kuper1}
Kupershmidt B.A., A super {K}orteweg--de {V}ries equation: an integrable
 system, \href{https://doi.org/10.1016/0375-9601(84)90693-5}{\textit{Phys.
 Lett.~A}} \textbf{102} (1984), 213--215.

\bibitem{kuper4}
Kupershmidt B.A., A review of superintegrable systems, in Nonlinear {S}ystems
 of {P}artial {D}ifferential {E}quations in {A}pplied {M}athematics,
 \textit{Lectures in Appl. Math.}, Vol.~23, American Mathematical Society,
 Providence, RI, 1986, 83--121.

\bibitem{lou}
Lou S.Y., Ren-integrable and ren-symmetric integrable systems,
 \href{https://doi.org/10.1088/1572-9494/ad23de}{\textit{Commun. Theor. Phys.
 (Beijing)}} \textbf{76} (2024), 035006, 8~pages,
 \href{http://arxiv.org/abs/2305.12388}{arXiv:2305.12388}.

\bibitem{yim}
Manin Yu.I., Gauge field theory and complex geometry, 2nd ed., \textit{Grundlehren Math.
 Wiss.}, Vol.~289,
 \href{https://doi.org/10.1007/978-3-662-07386-5}{Springer}, Berlin, 1997.

\bibitem{MR}
Manin Yu.I., Radul A.O., A supersymmetric extension of the
 {K}adomtsev--{P}etviashvili hierarchy,
 \href{https://doi.org/10.1007/BF01211044}{\textit{Comm. Math. Phys.}}
 \textbf{98} (1985), 65--77.

\bibitem{ml}
Mao H., Lv S.Q., On type-{II} B\"acklund transformation for the MKdV hierarchy,
 \href{https://doi.org/10.1515/zna-2016-0374}{\textit{Z.~Naturforsch.~A}}
 \textbf{72} (2017), 291--293.

\bibitem{mathieu}
Mathieu P., Supersymmetric extension of the {K}orteweg--de {V}ries equation,
 \href{https://doi.org/10.1063/1.528090}{\textit{J.~Math. Phys.}} \textbf{29}
 (1988), 2499--2506.

\bibitem{sNLS}
Roelofs G.H.M., Kersten P.H.M., Supersymmetric extensions of the nonlinear
 {S}chr\"odinger equation: symmetries and coverings,
 \href{https://doi.org/10.1063/1.529640}{\textit{J.~Math. Phys.}} \textbf{33}
 (1992), 2185--2206.

\bibitem{kai}
Tian K., Liu Q.P., A supersymmetric {S}awada--{K}otera equation,
 \href{https://doi.org/10.1016/j.physleta.2009.03.039}{\textit{Phys. Lett.~A}}
 \textbf{373} (2009), 1807--1810,
 \href{http://arxiv.org/abs/0802.4011}{arXiv:0802.4011}.

\bibitem{kw}
Tian K., Wang J.P., Symbolic representation and classification of~{$N=1$}
 supersymmetric evolutionary equations,
 \href{https://doi.org/10.1111/sapm.12163}{\textit{Stud. Appl. Math.}}
 \textbf{138} (2017), 467--498,
 \href{http://arxiv.org/abs/1607.03947}{arXiv:1607.03947}.

\bibitem{xl}
Xue L.-L., Liu Q.P., B\"acklund--{D}arboux transformations and discretizations
 of super {K}d{V} equation,
 \href{https://doi.org/10.3842/SIGMA.2014.045}{\textit{SIGMA}} \textbf{10}
 (2014), 045, 10~pages,
 \href{http://arxiv.org/abs/1312.6976}{arXiv:1312.6976}.

\bibitem{lz}
Zhang L., Zuo D., Integrable hierarchies related to the {K}uper-{CH} spectral
 problem, \href{https://doi.org/10.1063/1.3603817}{\textit{J.~Math. Phys.}}
 \textbf{52} (2011), 073503, 11~pages.

\bibitem{ztl}
Zhou H., Tian K., Li N., Four super integrable equations: nonlocal symmetries
 and applications,
 \href{https://doi.org/10.1088/1751-8121/ac6a2b}{\textit{J.~Phys.~A}}
 \textbf{55} (2022), 225207, 24~pages.

\bibitem{zty}
Zhou H., Tian K., Xiao Y., A super m{KdV} equation: bi-{H}amiltonian structures
 and {D}arboux transformations,
 \href{https://doi.org/10.1007/s12043-024-02737-y}{\textit{Pramana--J. Phys.}}
 \textbf{98} (2024), 52, 7~pages.

\bibitem{zhou}
Zhou R., A {D}arboux transformation of the~{${\rm sl}(2|1)$} super {K}d{V}
 hierarchy and a super lattice potential {K}d{V} equation,
 \href{https://doi.org/10.1016/j.physleta.2014.04.052}{\textit{Phys. Lett.~A}}
 \textbf{378} (2014), 1816--1819.

\bibitem{dfz}
Zuo D., Euler equations related to the generalized {N}eveu--{S}chwarz algebra,
 \href{https://doi.org/10.3842/SIGMA.2013.045}{\textit{SIGMA}} \textbf{9}
 (2013), 045, 12~pages,
 \href{http://arxiv.org/abs/1306.3628}{arXiv:1306.3628}.

\end{thebibliography}
\end{document}